\documentclass{article}
\usepackage[a4paper]{geometry}
\usepackage[sumlimits,namelimits]{amsmath} 
\usepackage{amssymb, amsthm}
\usepackage{mathabx}
\usepackage[english]{babel}
\usepackage{spverbatim}
\usepackage{lastpage}
\usepackage{tikz-cd}
\usetikzlibrary{arrows}
\usepackage{thmtools}
\usepackage[colorinlistoftodos,prependcaption,textsize=tiny]{todonotes}
\usepackage[T1]{fontenc}
\usepackage[utf8]{inputenc} 
\usepackage{authblk}
\usepackage{enumerate}
\usepackage{stmaryrd}
\usepackage[hidelinks]{hyperref}
\usepackage[capitalize]{cleveref} 
\usepackage{listings}
\usepackage{csquotes}
\usepackage[noend]{algpseudocode}  
\usepackage[toc,page]{appendix}
\usepackage[linesnumbered,rightnl,ruled]{algorithm2e}
\usepackage{sagetex}

\SetInd{2.5em}{0.25em}
\SetAlgoVlined
\SetKwProg{Try}{try}{:}{}
\SetKwProg{Catch}{catch}{:}{}
\SetKwFor{IfCustom}{if}{:}{}
\SetKwFor{ElseCustom}{else}{:}{}
\SetKwFor{ForCustom}{for}{:}{}
\SetKwFor{WhileCustom}{while}{:}{}
\SetKwComment{Comment}{/* }{ */}

\algnewcommand\algorithmicinput{\textbf{Complexity:}}
\algnewcommand\COMP{\item[\algorithmicinput]}
\SetKwComment{Comment}{//}{}
\definecolor{db}{RGB}{0, 153, 255}
\definecolor{dg}{RGB}{0, 153, 255}
\definecolor{dg}{RGB}{0, 19, 222}
\definecolor{dg}{RGB}{0, 204, 0}
\definecolor{dv}{RGB}{204, 0, 255}
\definecolor{dr}{RGB}{204, 0 , 0}
\definecolor{do}{RGB}{255,140,0}

\numberwithin{equation}{subsection} 

\fontencoding{T1}
\fontfamily{ptm}
\fontseries{m}
\fontshape{n}
\fontsize{12}{15}
\selectfont

\theoremstyle{definition}
\newtheorem{Definition}{Definition}[subsection]

\newtheorem{Remark}[Definition]{Remark}
\newtheorem{Hypothesis}[Definition]{Hypothesis}
\theoremstyle{plain}
\newtheorem{Corollary}[Definition]{Corollary}
\newtheorem{Lemma}[Definition]{Lemma}
\newtheorem{TheoremIntro}{Theorem}
\newtheorem{Theorem}[Definition]{Theorem}
\newtheorem{Proposition}[Definition]{Proposition}

\newcommand\FF{\mathbb{F}}
\newcommand\KK{\mathbb{K}}
\newcommand\LL{\mathbb{L}}
\newcommand\DD{\mathbb{D}}
\newcommand\can{K}
\newcommand\rr{L}
\newcommand\Pl{\mathcal{P}}
\newcommand\LAG{\operatorname{LAG}}
\newcommand\LD{\operatorname{LD}}
\usepackage{graphicx}
\newcommand{\hookuparrow}{\mathrel{\rotatebox[origin=c]{90}{$\hookrightarrow$}}}
\newcommand{\hookdownarrow}{\mathrel{\rotatebox[origin=c]{-90}{$\hookrightarrow$}}}
\newcommand{\End}{\operatorname{End}}
\newcommand{\Trd}{\operatorname{Trd}}
\newcommand{\Resrd}{\operatorname{Resrd}}
\newcommand{\res}{\operatorname{Res}}
\newcommand{\Res}{\operatorname{\mathcal{R}{es}}}
\newcommand{\ddesign}{d^*}
\newcommand{\dsrk}{d_{\mathrm{srk}}}
\newcommand{\wsrk}{w_{\mathrm{srk}}}
\newcommand{\rhoalgo}{\rho_{\mathrm{algo}}}
\newcommand{\Htpl}[1]{{\bf (H2-$#1$)}}

\algnewcommand{\algorithmicgoto}{\emph{go to}}
\algnewcommand{\Goto}[1]{\algorithmicgoto~\ref{#1}}

\makeatletter
\newcommand{\subjclass}[1]{
  \gdef\@subjclass{#1}
}
\newcommand{\keywords}[1]{
  \gdef\@keywords{#1}
}
\newcommand{\printclassifications}{
  \ifx\@subjclass\@empty\else
    \begingroup
      \renewcommand\thefootnote{}
      \footnotetext{\textit{2020 Mathematics Subject Classification.} \@subjclass}
    \endgroup
  \fi
  \ifx\@keywords\@empty\else
    \begingroup
      \renewcommand\thefootnote{}
      \footnotetext{\textit{Keywords:} \@keywords}
    \endgroup
  \fi
}
\makeatother

\newcommand{\ie}{\textit{i.e.,~}}


\title{Duality and decoding of\\linearized Algebraic Geometry codes}
\author{Elena Berardini}
\affil{CNRS, Mathematical Institute of Bordeaux, University of Bordeaux\\France\thanks{elena.berardini@math.u-bordeaux.fr}}
\author{Xavier Caruso}
\affil{CNRS, Mathematical Institute of Bordeaux, University of Bordeaux\\France\thanks{xavier.caruso@math.u-bordeaux.fr}}
\author{Fabrice Drain}
\affil{Mathematical Institute of Bordeaux\\University of Bordeaux\\France\thanks{fabrice.drain@math.u-bordeaux.fr}}
\date{\today}

\begin{document}

\maketitle

\printclassifications

\begin{abstract}
	We design a polynomial time decoding algorithm for linearized Algebraic Geometry codes with unramified evaluation places, a family of sum-rank metric evaluation codes on division algebras over function fields introduced in \cite{1}.
	By establishing a Serre duality and a Riemann--Roch theorem for these algebras, we prove that the dual codes of such linearized Algebraic Geometry codes, that we term linearized Differential codes, coincide with the linearized Algebraic Geometry codes themselves over the adjoint algebra, and that our decoding algorithm is correct.
\end{abstract}
\tableofcontents
\section*{Introduction}

Sum-rank metric codes have seen a resurgence of interest in recent decades due to their broad applications in areas like network coding, distributed data storage, and code-based cryptography as described in \cite{13}.
Recent code constructions, such as the linearized Reed--Solomon codes introduced in \cite{12}, achieve the sum-rank Singleton-type bound.
Furthermore, linearized Algebraic Geometry (LAG) codes were introduced and studied by the first and second authors in~\cite{1}, using Riemann--Roch spaces on some division algebras over function fields of curves.
They generalize both Algebraic Geometry (AG) codes as described in~\cite{2} and linearized Reed--Solomon codes.
This lends them good properties in the sum-rank metric: they almost meet the sum-rank Singleton bound and their length is not limited by the size of the alphabet.
We also mention that other algebraic constructions exist, such as linearized Reed--Muller codes, described in~\cite{31}.

The main contributions of this paper are to study the duality properties of linearized AG codes and to design a decoding algorithm for them, under the hypothesis that the evaluation places used in their definition are unramified.

We briefly outline the construction of linearized AG codes (a more detailed presentation is given at the beginning of Section~\ref{LAGCODEDUALITYANDDECODING} of the present paper).
Let $\KK$ be a the function field of an algebraic curve defined over a field $\FF$.
We fix a cyclic extension $\LL/\KK$ of degree $r$ and a generator $\theta$ of its Galois group. We let $\LL[T;\theta]$ be the corresponding Ore polynomial ring.
We set $$\DD:=\frac{\LL[T;\theta]}{T^r-x},$$ where $x \in \KK$ is chosen so that $\DD$ is a division algebra as in \cite[Section IV A]{1}.

In~\cite{1}, valuations on $\DD$ (indexed by the places of $\LL$) are introduced and used to define Riemann--Roch spaces $\Lambda_\DD(A) \subset \DD$ associated to certain rational divisors $A$ of $\LL$.
Moreover, for any rational place $P$ of $\KK$ away from $A$ (meaning that the support of $A$ does not contain any place of $\LL$ above $P$) and satisfying a certain hypothesis called \Htpl{P_i} in \cite{1}, an evaluation morphism
$\bar\varepsilon_P : \Lambda_\DD(A) \to \operatorname{End}_\FF(V_P)$ is constructed, 
where $V_P$ is the residual algebra of $\LL$ at $P$.
Fixing a finite set $\{P_1, \ldots, P_s\}$ of such places and writing $D = P_1 + \cdots + P_n$, one can concatenate the $\bar\varepsilon_i$ to obtain a multievaluation morphism
$$\bar\varepsilon_{D} : \Lambda_{\DD}(A) \rightarrow \prod_{i=1}^s\operatorname{End}_{\FF}\left(V_{P_i}\right)\simeq \big(\FF^{r \times r}\big)^s,$$
whose image is by definition the linearized AG code $\LAG_\DD(A, D)$.
It is endowed with the sum-rank metric defined by the weight function $(M_1, \ldots, M_s) \mapsto \sum_{i=1}^s \operatorname{rank}(M_i)$.

Theorem~2 of \cite{1} provides estimates on the parameters of the code $\operatorname{LAG}_{\DD}(A,D)$:
assuming $\operatorname{deg}(A) < sr$, its minimum distance is at least $\ddesign(A) := sr-\operatorname{deg}A$, 
while its dimension is that of the Riemann--Roch space $\Lambda_\DD(A)$,
which is itself lower bounded by $r\deg(A) - r (g-1)$, where $g$ is given by an explicit formula (see Definition \ref{def:genus}) and is called the \emph{genus} of $\DD$. 

In this paper, we propose a decoding algorithm for linearized AG codes (see Algorithm \ref{ALGO}) and prove the following theorem.

\begin{TheoremIntro}[Theorem \ref{LAGDECODING}]
	Let $A$ be a divisor as above, and satisfying  $2g -2< \operatorname{deg}A < sr$. We assume the existence of a non-split non-evaluation rational place of $\KK$ (see Hypothesis \ref{HYPO2}).
	Set $$\rhoalgo := \left\lfloor \frac{\ddesign(A) -g-1}{2}\right\rfloor.$$
	Then, our algorithm decodes any error relative to the code $\operatorname{LAG}_{\DD}(A,D)$ of sum-rank weight at most $\rhoalgo$, with polynomial complexity in $s$ and $r$.\end{TheoremIntro}
Our decoding algorithm is the analogue in the sum-rank metric to the classical decoding algorithm for AG codes in the Hamming metric described in \cite[Section 8.5]{2}.
We express the received word $m$ as the sum $m = c+e$ of the original codeword $c \in \operatorname{LAG}_{\DD}(A,D)$ and an error $e$ of sum-rank weight at most $\rhoalgo$.
We first determine a space $U$ localizing the error $e$ with sufficient precision, \ie satisfying $\operatorname{im}(e) \subset U$
and $\operatorname{dim}(U) < \ddesign(A)$.
Then, we solve the syndrome equations for $e$ restricted to the localizing space.

The fundamental theorem underlying our decoding algorithm is a Riemann--Roch Theorem on $\DD$ for what we call \emph{extended divisors}.
Those are pairs composed of a divisor $A$ away from $D$ and a product of vector subspaces $W := \bigoplus_{i=1}^{s} W_{i}$, where $W_{i} \subset V_{P_i}$.
To an extended divisor $(A, W)$, we associate its degree $\deg(A, W) = \deg A - \dim_\FF W$ and the Riemann--Roch space
$$\Lambda_{\DD}(A,W):=\left\{f \in \Lambda_{\DD}(A) \mid \bar\varepsilon_i(f)_{|W_{i}}=0 \quad \forall i \in \{1,\ldots,s\} \right\}.$$
We denote its dimension by $\lambda_{\DD}(A,W)$.
Besides, we define a canonical divisor $K$ related to $\DD$, whose degree is $2g - 2$.
Letting $\DD^\star = \frac{L[T; \theta]}{T^r - x^{-1}},$
our Riemann--Roch Theorem reads as follows.

\begin{TheoremIntro}[Theorem \ref{RROCHTHEO}]
	Let $(A,W)$ be an extended divisor. We have
	$$\lambda_{\DD}(A,W) = r\operatorname{deg}(A,W)+r(1-g)+\lambda_{\DD^\star}\big(\can + \operatorname{CoNr}(D) - A, W^\perp\big)$$
	where $\operatorname{CoNr}(D) = \sum_{i=1}^s \sum_{Q\mid P_i} Q$ and $W^\perp$ is the orthogonal of $W$ with respect to the trace bilinear form
	(see Subsection~\ref{RROmega}).

	Moreover, if $\operatorname{deg}(A,W)<0$, then we have $\lambda_{\DD}(A,W)=0$.
\end{TheoremIntro}

We show this theorem by extending the proof of the classical Riemann--Roch Theorem as presented in \cite[Chapter 1]{2} and in \cite[Chapter 7]{8}.
As an important ingredient of the proof, we define adele spaces associated to $\DD$ and differential spaces associated to $\DD^\star$, and
we show a Serre-type duality theorem, establishing a perfect pairing between them.
As in the classical case, this theorem brings along a canonical duality between linearized AG codes
and some linearized Differential codes defined by means of residues (see Definition~\ref{CODEDIFF}).

\begin{TheoremIntro}[Theorem~\ref{NONCAN}]
	For any divisor $A$ away from $D$ and a particular choice of $K$, we have the following equalities of codes:
	$$\operatorname{LAG}_{\DD}(A,D)^\perp  = \operatorname{LD}_{\DD^\star}(A,D) = \operatorname{LAG}_{\DD^\star}({\can+\operatorname{CoNr}(D)-A},D).$$
\end{TheoremIntro}
We note that a similar theorem was proved in~\cite{34} for linearlized Reed--Solomon codes,
which correspond to the case where $\KK=\mathbb{F}_q(t)$, $\LL=\mathbb{F}_{q^r}(t)$ and $A$ is a scalar multiple of the point at infinity.
On the other hand, in this situation, the result of~\cite{34} encompasses a more general setting, as it includes Ore polynomials with derivations, and restrictions of the morphisms $\bar\varepsilon_D$ to subspaces.

\paragraph{Organisation of the paper.} 
The paper is structured as follows. In Section \ref{RRTHEOSEC}, we establish our Riemann--Roch Theorem.
In Section \ref{LAGCODEDUALITYANDDECODING}, we move to applications to coding theory, establishing first the duality theory for linearized AG codes, and then designing their decoding algorithm. Besides, in Appendix~\ref{SAGE}, we provide some results from a SageMath implementation of our algorithm.

\section{Riemann--Roch theorem on division algebras}\label{RRTHEOSEC}

We consider the setting of \cite{1}.
Let $\FF$ be an arbitrary field.
Let $\KK$ be an algebraic function field in one variable over its constant field $\FF$.
Let $\LL/\KK$ be a cyclic extension of degree $r \geq 1$.
We denote by $\theta$ a generator of its cyclic Galois group.

We denote by $\Pl_{\KK}$ and $\Pl_{\LL}$ the sets of places of $\KK$ and $\LL$, respectively. 
For a place $P \in \Pl_{\KK}$ (resp.~$Q \in \Pl_{\LL}$), we write $\KK_P$ (resp.~$\LL_Q$) for the completion of $\KK$ (resp.~$\LL$) at that place.
We also use $\LL_P$ (resp. $\DD_P$) to denote the completion of $\LL$ (resp. $\DD$) at a place $P \in \Pl_{\KK}$. 
We denote by $v_P$ (resp. $w_Q$) the normalized valuation of $\KK$ (resp. of $\LL$) at a place $P$ (resp. $Q$). We use the notation $Q|P$ to indicate that $Q$ is a place of $\LL$ above a place $P$ of $\KK$. 
Writing $e_P$ the ramification index of $Q|P$, we have $w_Q = e_P v_P$. 
We also use the notation $\mathcal{O}_{P}$ (resp.~$\mathcal{O}_{Q}$) for the valuation ring of $\KK_P$ (resp.~$\LL_Q$) and set 
$$\mathcal{O}_{\LL,P}:= \prod_{Q|P}\mathcal{O}_{Q}\subset \LL_P.$$

We consider the Ore polynomial ring $\LL[T;\theta]$, defined in \cite{24} as polynomials equipped with a twisted multiplication given by the rule $Ta= \theta(a)T$ for any $a \in \LL$.
We note that the ring $\LL[T;\theta]$ is left and right Euclidean as shown in \cite[Section 2]{24}.
Let $x \in \KK$. We define as in \cite{1} a cyclic algebra $\DD$ by
$$\DD:=\frac{\LL[T;\theta]}{T^r-x}.$$
Throughout the paper, we assume that $x$ is chosen in such a way that $\DD$ is a division algebra (see Hypothesis \textbf{(H1)} in \cite[Section IV]{1}).
For $P \in \Pl_\KK$, we set $\DD_P = \KK_P \otimes_\KK \DD$.

In the sequel, we will tacitly use the obvious isomorphisms
$$\begin{array}{rclcrcl}
  \DD & \tilde{\longrightarrow} & \LL^r & \quad\text{and}\quad & 
  \LL \otimes_{\KK} \KK_P & \tilde{\longrightarrow} & \prod_{Q|P} \LL_Q \\
  \sum_{i=0}^{r-1}\alpha_i T^i & \mapsto & (\alpha_i)_{i=0}^{r-1} & &
  y & \mapsto & (y_Q)_{Q|P}
\end{array}.$$
We will denote by $y_i \in \LL$ (resp.~$y_i \in \prod_{P\in\Pl_{\KK}}\LL\otimes_{\KK} \KK_{P}$) the $i$-th component of $y \in \DD$ (resp.~$y\in \prod_{P\in\Pl_{\KK}}\DD\otimes_{\KK} \KK_{P}$), induced by the first isomorphism, and by $y_Q \in \LL_Q$ the $Q$-th component of $y \in \LL \otimes_{\KK} \KK_P$ for $Q|P$, induced by the second isomorphism.

\begin{Remark}
	As shown in \cite[Theorem 32.20]{11}, 
	if $\FF$ is a finite field, all finite dimensional central simple algebras are cyclic and can be represented in the form $\frac{\LL[T;\theta]}{T^r-x}$. Therefore, in the context of application to linear codes over algebraic curves defined over a finite field of constants, our choice of $\DD$ is general. 
	However, for an arbitrary constant field $\FF$ or for a higher dimensional function field $\KK$, the decomposition $\frac{\LL[T;\theta]}{T^r-x}$ does not represent all central simple algebras over $\KK$ anymore.
\end{Remark}

The goal of this section is to prove a Riemann--Roch Theorem on $\DD$. Beyond its own interest, this result will be pivotal to design our decoding algorithm for linearized AG codes.

We mention that, following the Ph.D. thesis of Witt \cite{25}, Tamagawa \cite{14} and later Mattson \cite{28} showed a Riemann--Roch Theorem in the more general setting of central simple algebras over a function field.
However, our approach based on the decomposition $\DD=\frac{\LL[T;\theta]}{T^r-x}$, is more explicit and allows for effective duality with differential forms. In particular, this will be pivotal, in Subsection \ref{LAGCODEDUALITY}, to canonically characterize the dual of a linearized AG code as a linearized Differential code.

\subsection{The case of numerical divisors}\label{RRTHEO}

For any place $P\in \Pl_{\KK}$, we define the positive integer
$$b_{P}:=\frac{r}{\operatorname{gcd}(e_P v_P(x),r)}.$$
We note that we have $b_{P}=1$ at all places where $v_{P}(x)=0$, hence almost everywhere.

We define the group of \emph{divisors} of $\DD$ 
$$\operatorname{Div}(\DD) := \bigoplus_{P \in \Pl_{\KK}}\bigoplus_{Q | P} \frac{\mathbb{Z}}{b_{P}} Q.$$

The \emph{degree} of a divisor $A = \sum_Q A_Q Q$ is, by definition, $\deg A := \sum_Q A_Q\deg Q$.

We define the following $Q$-functions $\omega_{Q}$
$$\begin{array}{rcl}
\omega_{Q}: \quad
		\DD_P \simeq \sum_{j=0}^{r-1} \LL_P T^j     & \longrightarrow & \frac{\mathbb{Z}}{b_{P}}                                               \\
		y=\sum\limits_{j=0}^{r-1}({y_j}_{Q'})_{Q'|P} T^j & \mapsto & \min\limits_{0 \leq j <r} \left(w_Q({y_j}_Q) + j \frac{w_{Q}(x)}{r}\right)
	\end{array}$$
	and the gauge $$\begin{array}{rcl} \omega_P : \quad
		\DD_P     & \longrightarrow & \frac{\mathbb{Z}}{b_{P}e_P}                       \\
		(\alpha_Q)_{Q|P} & \mapsto & \min\limits_{Q|P}{\frac{\omega_{Q}({\alpha_Q})}{e_P}}
	\end{array}.$$
Note that the definition $\omega_{Q}$ corresponds to the definition $e_P \cdot w_{j,x}$ from \cite{1}.
\begin{Lemma}\label{SURMULT}
	The gauge $\omega_{P}$ is surmultiplicative, that is, for every $f,g \in \DD_P$ it holds $$ \omega_{P}(fg) \geq  \omega_{P}(f) + \omega_{P}(g).$$
\end{Lemma}
\begin{proof}
	Let $P \in \Pl_{\KK}$. For $f,g \in \DD_P$, we have
	$$\omega_P(fg)=\omega_P\left(\sum_{i=0}^r f_{i}T^{i} \sum_{i=0}^r g_{i} T^{i}\right)= \omega_P\left(\sum_{i=0}^r \sum_{k=0}^i f_{k}\theta^{k}(g_{i-k}) T^{i} + \sum_{i=r}^{2r-1} \sum_{k=i-r}^r f_{k}\theta^{k}(g_{i-k}) T^{i-r} \right).$$
	Let $i_0$ be the index of a minimal summand.
	We let  $Q_1, \ldots, Q_{m_P}$ be the places above $P$, indexed such that $\Phi(Q_j) = Q_{(j+1) \% m_P}$.
	We distinguish the two cases $i_0 \geq r$ and $0\leq i_0 < r$.
	In the first case, we have $$\begin{aligned} \omega_P(fg)&\geq \min_j\left(\frac{w_{Q_j}\left(\sum_{0\leq k \leq i_0} f_{k}\theta^{k}(g_{i_0-k})   \right)}{e_P}+\frac{(i_0-r) v_P(x)}{r}\right)  \\
        &\geq \min_j \min_{0\leq k \leq i_0}\left(\frac{w_{Q_j}\left(x f_{k}\theta^{k}(g_{i_0-k})\right)}{e_P}+\frac{(i_0-r) v_P(x)}{r}\right) \\
	& \geq \min_j \min_{0\leq k \leq i_0}\left(\frac{w_{Q_j}(f_{k})}{e_P}+\frac{k v_P(x)}{r}+\frac{w_{Q_{(j+k)\%m_P}}(g_{i_0-k})}{e_P}+\frac{(i_0-k)  v_P(x)}{r}\right),\end{aligned}$$
	which is greater than $ \omega_P(f) +\omega_P(g)$.
	In the second case, we have $$\begin{aligned} \omega_P(fg) &\geq \min_j\left(\frac{w_{Q_j}\left(\sum_{0\leq k \leq i_0} f_{k}\theta^{k}(g_{i_0-k})  \right)}{e_P}+\frac{i_0  v_P(x)}{r}\right) \\
        &\geq \min_j \min_{0\leq k \leq i_0}\left(\frac{w_{Q_j}\left(f_{k}\theta^{k}(g_{i_0-k})\right)}{e_P}+\frac{i_0  v_P(x)}{r}\right)\\
	&\geq \min_j \min_{0\leq k \leq i_0}\left(\frac{w_{Q_j}(f_{k})}{e_P}+\frac{k v_P(x)}{r}+\frac{w_{Q_{(j+k)\%m_P}}(g_{i_0-k})}{e_P}+\frac{(i_0-k)  v_P(x)}{r}\right),\end{aligned}$$
	which is also greater than $ \omega_P(f) +\omega_P(g)$.
\end{proof}

Using the surmultiplicativity property of the gauge from Lemma \ref{SURMULT}, we can state the following definition.
\begin{Definition}\label{def:gaugeorders}
	For any place $P\in\Pl_{\KK}$, we define the order $\Lambda_P$ as  $$\Lambda_P:=\{f \in \DD_P \mid \omega_P(f)\geq 0 \}.$$
\end{Definition}

\begin{Definition}
	The \emph{principal divisor} of a function $f \in \DD$ is
	$(f) = \sum_{Q \in \Pl_{\LL}} \omega_Q(f) Q.$
\end{Definition}

\begin{Definition}[\protect{\cite[Definition 2.3]{1}}]
	For any $A \in \operatorname{Div}(\DD)$, the \emph{Riemann--Roch space} associated to $A$ is 
	$$\Lambda_{\DD}(A):=\left\{f \in \DD \mid \omega_{Q}(f)+A_{Q} \geq 0 \quad \forall Q \in \Pl_{\LL}  \right\}.$$
	We denote by $\lambda_{\DD}(A)$ its  dimension. 
\end{Definition}

The aim of our Riemann--Roch Theorem is to relate $\lambda_{\DD}(A)$ to the dimension of another Riemann--Roch space.
We now introduce the necessary notions to define the parameters corresponding to this second Riemann--Roch space.

\begin{Definition}\label{DELTADEF}
	The \emph{different divisor} $\Delta$ of $\DD$ (relative to the gauges $\omega_P$) is
        $$\Delta:=\sum\limits_{P \in \Pl_{\KK}}\sum\limits_{Q|P} \frac{b_{P}-1}{b_{P}} Q.$$
\end{Definition}

\begin{Definition}\label{DEFK}
	A \emph{canonical divisor of $\DD$} is
	$\can := \can_\LL+ \Delta$
        where $\can_\LL$ is a canonical divisor of~$\LL$.
\end{Definition}
\begin{Definition}\label{def:genus}
	The \emph{genus} of $\DD$ (relative to the gauges $\omega_P$) is the rational number $g$ defined by $2g - 2 =\operatorname{deg}(\can)$.
\end{Definition}
The genus $g$ does not depend on the choice of $\can$ as $\operatorname{deg}(\can_{\LL})$ does not either.

\begin{Definition}\label{DEFADJOINTALGEBRA}
	The \emph{adjoint algebra} of $\DD$ is
	$${\DD}^\star := \frac{\LL[{T};\theta]}{{T}^r-x^{-1}}.$$
\end{Definition}
\begin{Definition}\label{DEFADJUNCTION}
	We define the \emph{adjunction} on $\DD$ as the following involution: $$\begin{array}{rcl}
			\DD & \longrightarrow & {\DD}^\star \\
			f=\sum\limits_{i=0}^{r-1}\alpha_i T^i  & \mapsto & {f}^{\star}= \sum\limits_{i=0}^{r-1} {T}^{-i}\alpha_i
                        \quad \text{(with } \alpha_i \in \LL\text{)}.
		\end{array}$$ 
\end{Definition}
We readily check that $f \mapsto f^\star$ is an anti-isomorphism, \ie it satisfies $(fg)^\star = g^\star f^\star$, and
that $\omega_P(f) = \omega_P(f^\star)$ for any place $P \in \Pl_\KK$.
We define the maps
	$$\begin{array}{rclcrcl}
			\operatorname{Div}(\DD)              & \stackrel{\pi_i}{\longrightarrow} & \operatorname{Div}(\LL) & \quad ; \quad & 
			\operatorname{Div}(\DD^\star)        & \stackrel{\pi^\star_i}{\longrightarrow} & \operatorname{Div}(\LL)  \\
			\sum\limits_{Q \in \Pl_{\LL}} A_{Q} Q & \mapsto & \sum\limits_{Q \in \Pl_{\LL}}\lfloor A_{Q}+i \frac{w_{Q}(x)}{r} \rfloor Q & &
			\sum\limits_{Q \in \Pl_{\LL}} A_{Q} Q & \mapsto & \sum\limits_{Q \in \Pl_{\LL}}\lfloor A_{Q}-i \frac{w_{Q}(x)}{r} \rfloor Q
		\end{array}$$
		where $\operatorname{Div}(\DD^\star):=\operatorname{Div}(\DD)$. 
In the sequel, we shall use ${\pi_i}$ and ${\pi^\star_i}$ in combination with the isomorphisms below as in \cite[Equation (5)]{1}, and write
\begin{equation}\label{ISOLBDA}
	\Lambda_{\DD}(A) \simeq \oplus_{i=0}^{r-1}L(\pi_i(A)) \quad \text{ and }  \quad \Lambda_{\DD^\star}(A)\simeq \oplus_{i=0}^{r-1}L(\pi^\star_i(A)),
\end{equation}
where $L(-)$ denotes the classical Riemann--Roch space of $\LL$ associated to the divisor in argument.

\begin{Theorem}[Riemann--Roch Theorem]\label{ELEM}
	For any $A \in \operatorname{Div}(\DD)$ and any choice of a canonical divisor $\can$, we have
	\begin{equation*}\lambda_{\DD}(A)=r \cdot \operatorname{deg}{A}+r(1-g) +\lambda_{\DD^\star}({\can-A}).\end{equation*}
\end{Theorem}
\begin{proof}\label{proof}
	Applying the classical Riemann--Roch theorem on $\LL$ \cite[Theorem 1.5.15]{2}, we have for any choice of $\can_\LL$ and for any $i \in \{0,\ldots,r-1\} $
	\begin{equation}\label{Ri} \operatorname{dim} \rr\left(\pi_i(A)\right)=\operatorname{deg}\pi_i(A)+1-g_{\LL}+\operatorname{dim} \rr\left(\can_\LL-\pi_i(A)\right).\end{equation}
	We sum the Equations \eqref{Ri} over the indices $i$ and consider each summand of the equation separately.
	First, using Equation \eqref{ISOLBDA}, we get $\sum_{i=0}^{r-1}\operatorname{dim}\rr(\pi_i(A))=\lambda_{\DD}(A)$.
		Secondly, from \cite[Lemma 2.5]{1}, we have $\sum_{i=0}^{r-1} \operatorname{deg}(\pi_i(A)) = r\operatorname{deg}A -r\frac{\operatorname{deg}(\Delta)}{2}$. Thirdly, since $g_{\LL}=g - \frac{\operatorname{deg}{\Delta}}{2}$,  we have $\sum_{i=0}^{r-1} (1-g_{\LL}) = r(1-g)+r\frac{\operatorname{deg}(\Delta)}{2}$.
Finally, we set $a = A_{Q}+i \frac{w_{Q}(x)}{r}$. As $a$ belongs to $\frac{\mathbb{Z}}{b_P}$, we get $\left\lfloor\frac{b_{P}-1}{b_{P}}-(a-\lfloor a\rfloor)\right\rfloor =0  $.
		      Now, summing $-\lfloor a\rfloor =\left\lfloor\frac{b_{P}-1}{b_{P}}-a\right\rfloor$ over all places $Q \in \Pl_{\LL}$, we get $-\pi_i(A)=\pi_i^\star(\Delta-A)$.
		      Thus, we get $\can_\LL-\pi_i(A)= \can_\LL+\pi_i^\star(\Delta-A)$. 
			  As $\can_\LL$ belongs to $\operatorname{Div}(\LL) \subset \operatorname{Div}(\DD)$, we have 
			  $\can_\LL+ \pi_i^\star(\Delta-A)=\pi_i^\star({\can-A})$.
		      Hence, using Equation \eqref{ISOLBDA}, we retrieve $\sum_{i=0}^{r-1}\operatorname{dim}\rr\left(\can_\LL-\pi_i(A)\right)=\sum_{i=0}^{r-1}\operatorname{dim}\rr(\pi^\star_i({\can-A}))=\lambda_{\DD^\star}({\can-A})$. 

	Summing the four terms, we recover the statement.
\end{proof}

\subsection{Riemann theorem for extended divisors} \label{RAPPEL}

We now aim at proving an extension of Theorem~\ref{ELEM} allowing for more general divisors, that we call ``extended divisors'' (see Definition \ref{EXTDIV} below).

For this, we fix $\{P_1,\ldots,P_s\}$, a finite set of $s$ rational \emph{unramified places} of $\KK$ indexed by $[s]:=\{1,\ldots,s\}$
and set $D:=P_1+\cdots+P_s \in \operatorname{Div}(\KK)$.
We assume moreover that all $P_i$ satisfy the hypothesis \Htpl{P_i} of \cite[Section IV A]{1}, \ie that there exist elements $u_Q \in \LL_Q^\times$ such that, for all $Q \in \Pl_{\LL}$ above $P_i$, we have
	\begin{equation}
		w_{Q}\left(u_{Q}\right) =\frac{w_{Q}(x)}{r}
                \quad \text{and} \quad
		x                       =\left(\operatorname{Nr}_{\LL_{Q} / \KK_{P_i}}\left(u_{Q}\right)\right)_{Q|P} \in \LL_P. \label{H2PREC}
	\end{equation}
If we assume $\FF$ to be finite, a sufficient condition for this to hold is that $r$ divides $v_{P_i}(x)$ (see \cite[Lemma 6]{1}).

We set $u_{P_i} :=(u_Q)_{Q|P_i} \in \LL_{P_i}$ and recall from~\cite[Section IV A]{1} that there is an isomorphism of $\KK_{P_i}$-algebras
$\varepsilon_{i} : \DD_{P_i} \to \operatorname{End}_{\KK_{P_i}}(\LL_{P_i})$, $f(T) \mapsto f(u_{P_i}\theta)$.
Moreover, its restriction to $\Lambda_{P_i}$ induces a second isomorphism
\begin{equation} \varepsilon_{i|\Lambda_{P_i}}:  \Lambda_{P_i} \tilde{\longrightarrow} \operatorname{End}_{\mathcal{O}_{P_i}}(\mathcal{O}_{\LL,{P_i}}).\end{equation}
For $i\in[s]$, we define the \emph{residual algebra} at $P$ by $V_i:= \mathcal{O}_{\LL,P_i}/P_i\mathcal{O}_{\LL,P_i}$ and set
$V_D:= \oplus_{i=1}^s V_{i}$.
Since $P_i$ is a rational place, we have $\dim_\FF V_i = r$ and so $\operatorname{dim}_{\FF} V_D= sr$.

	We denote by $\operatorname{Div}(\DD,D)$ the subgroup of divisors of $\DD$ away from $D$, that~is, $$\operatorname{Div}(\DD,D):=\{A \in \operatorname{Div}(\DD) \mid A_Q =0 \quad \forall Q | P_i \quad \forall i \in  [s]\}.$$

\begin{Definition}	\label{EXTDIV}
	An \emph{extended divisor} $(A,W)$ is a pair composed of a divisor $A \in \operatorname{Div}(\DD,D)$ and a direct sum of $\FF$-vector subspaces $W :=\bigoplus_{i=1}^s W_{i} \subset V_{D}$, where $W_{i} \subset V_{i}$.

	Its \emph{degree} is $\operatorname{deg}(A,W):=\operatorname{deg}A-\operatorname{dim}_{\FF}(W)$.
\end{Definition}
We recall the definition of the space of adeles $\mathcal{A}_{\LL}(A)$ of $\LL$ for a divisor $A= \sum_{Q \in \Pl_{\LL}} A_Q Q \in \operatorname{Div}(\LL)$ as in \cite[Definition 1.5.3]{2}:
$$\mathcal{A}_{\LL}(A):= \left\{(y_Q)_Q \in  \prod\limits_{Q \in \Pl_{\LL}}\LL_Q \, \middle\vert\, w_{Q}({y}_Q)+A_Q \geq 0  \right\}.$$
We also let $\mathcal{A}_{\LL}$ be the union of all $\mathcal{A}_{\LL}(A)$. It is also the set of sequences $(y_Q)_Q$ such that $w_{Q}({y}_Q) \geq 0$ for almost all $Q$.

\begin{Definition}
        For a divisor $A= \sum_{Q \in \Pl_{\LL}} A_Q Q \in \operatorname{Div}(\DD)$, we define
	$$\mathcal{A}_{\DD}(A)  :=\left\{(\alpha_P)_P \in \prod\limits_{P \in \Pl_{\KK}}\DD_{P}   \, \middle\vert\, \omega_{Q}({\alpha}_P) + A_Q \geq 0 \quad  \forall P \quad  \forall Q | P\right\}$$
        and let $ \mathcal{A}_{\DD}$ be the union of all $\mathcal{A}_{\DD}(A)$.
\end{Definition}
Note that by definition, we have $\Lambda_{\DD}(A)= \mathcal{A}_{\DD}(A) \cap \DD$.
\begin{Lemma}\label{TRIVISO}
	We have the following isomorphism of $\mathcal{A}_\KK$-algebras $$\mathcal{A}_{\DD} \simeq \oplus_{i=0}^{r-1}\mathcal{A}_{\LL}T^i.$$
\end{Lemma}
\begin{proof}
	The support of $(x)$ in $\Pl_{\LL}$ being finite,
	we conclude using the isomorphism $\LL\otimes_{\KK} \KK_P \simeq \prod_{Q|P}\LL_Q $.
\end{proof}
In the sequel, we will use that $\LL$ (resp.~$\DD$) injects diagonally in $\mathcal{A}_{\LL}$ (resp.~$\mathcal{A}_{\DD}$).
\begin{Lemma}\label{ADELEVAL}
	Let $A \in \operatorname{Div}(\DD)$.
	We have the following isomorphisms of $\FF$-vector spaces 
        $$\mathcal{A}_{\DD}(A) \simeq \oplus_{i=0}^{r-1}\mathcal{A}_{\LL}(\pi_i(A))
        \quad\text{and}\quad
        \frac{\mathcal{A}_{\DD}}{\mathcal{A}_{\DD}(A) + \DD}  \simeq \oplus_{i=0}^{r-1}\frac{\mathcal{A}_{\LL}}{ \mathcal{A}_{\LL}(\pi_i(A))+\LL}.$$
\end{Lemma}
\begin{proof}
	The first isomorphism is analogue to the isomorphism of Equation \eqref{ISOLBDA}.
	The isomorphism of Lemma \ref{TRIVISO}  and the isomorphism $\mathcal{A}_{\DD}(A) + \DD \simeq \oplus_{i=0}^{r-1}(\mathcal{A}_{\LL}(\pi_i(A))+\LL) $ together with the natural injection $\mathcal{A}_{\DD}(A) + \DD \hookrightarrow \mathcal{A}_{\DD}$ induced by the diagonal embedding of $\DD$
	in $\mathcal{A}_{\DD}$ allow us to build the quotients and prove the second isomorphism. 
\end{proof}
We recall the definition of the evaluation map at $P_i \in \Pl_{\KK}$ for any divisor $A \in \operatorname{Div}(\DD,D)$ as the following composition of morphisms according to \cite[Definition 3.3]{1}:
\begin{equation}\label{EPS}
\bar\varepsilon_i: \Lambda_{\DD}(A) \hookrightarrow  \Lambda_{P_i} \overset{\varepsilon_{i}}{\tilde{\longrightarrow}} \operatorname{End}_{\mathcal{O}_{{P_i}}}(\mathcal{O}_{\LL,{P_i}}) \twoheadrightarrow \operatorname{End}_{\FF}(V_{i}),
\end{equation}
where the first map is the inclusion and the last one is the projection.

\begin{Definition}\label{RREXT}
	Let $(A,W)$ be an extended divisor. We define:
        \begin{align*}
	\mathcal{A}_{\DD}(A,W) & :=\left\{f \in \mathcal{A}_{\DD}(A)\mid \bar\varepsilon_i( f )_{|W_{i}}=0 \quad \forall i \in [s]\right\}, \\
        \Lambda_{\DD}(A,W) & :=\mathcal{A}_{\DD}(A,W) \cap \DD.
        \end{align*}
	We denote by $\lambda_{\DD}(A,W)$ the dimension of $\Lambda_{\DD}(A,W)$.
\end{Definition}

	We define the conorm of a divisor as follows:
	$$\begin{array}{rcl}
                        \operatorname{CoNr}: \quad \operatorname{Div}(\KK) & \rightarrow & \operatorname{Div}(\LL) \subset  \operatorname{Div}(\DD) \\
			P                     & \mapsto & e_P \sum_{Q|P}  Q
		\end{array}.$$

\begin{Lemma} \label{VANISHING}
	For $A \in \operatorname{Div}(\DD,D)$, we have $\mathcal{A}_{\DD}(A,V_D) = \mathcal{A}_{\DD}(A-\operatorname{CoNr}(D)).$
\end{Lemma}
\begin{proof}
	By definition, the kernel of  
	$\varepsilon_i: \Lambda_{P_i} \overset{\varepsilon_{i}}{\tilde{\longrightarrow}} \operatorname{End}_{\mathcal{O}_{{P_i}}}(\mathcal{O}_{\LL,{P_i}}) \twoheadrightarrow \operatorname{End}_{\FF}(V_{i})$ is 
	 the inverse image of $\operatorname{Hom}_{\mathcal{O}_{{P_i}}}(\mathcal{O}_{\LL,{P_i}},P_i\mathcal{O}_{\LL,{P_i}})$ by $\varepsilon_{i}$, which is $\mathcal{A}_{\DD}(A-\operatorname{CoNr}(D))_{P_i}$.
\end{proof}

\begin{Definition}
	The \emph{index of speciality} of an extended divisor $(A,W)$ is
	$$i(A,W):=\lambda_{\DD}(A,W)-r\operatorname{deg}(A,W)-r(1-g).$$
\end{Definition}

\begin{Proposition}\label{ROCH}
	For any extended divisor $(A,W)$, we have
        $$i(A,W)=\operatorname{dim}_{\FF} \left(\frac{\mathcal{A}_{\DD}}{\mathcal{A}_{\DD}(A,W)+\DD}\right).$$
\end{Proposition}
\begin{proof}
	We have an isomorphism
        $$\begin{array}{rcl}
	\oplus_{i=0}^{r-1}\operatorname{Hom}\left(\frac{\mathcal{A}_{\LL} }{\mathcal{A}_{\LL}(\pi_i(A)) + \LL},\FF\right) & \xrightarrow{\sim} & \operatorname{Hom}\left(\frac{\mathcal{A}_{\DD} }{\mathcal{A}_{\DD}(A) + \DD},\FF\right) \\
	(\nu_0, \ldots, \nu_{r-1}) & \mapsto & (\alpha \mapsto \sum_{i=0}^{r-1}\nu_i(\alpha_i))
	\end{array}.$$
        Indeed, the map is obviously injective and 
	it is also surjective, a preimage of $\nu \in  \operatorname{Hom}\left(\frac{\mathcal{A}_{\DD} }{\mathcal{A}_{\DD}(A) + \DD},\FF\right)$ being given by of the linear forms $\nu_i : \alpha_i \mapsto \nu((0,\ldots,\alpha_i,\ldots,0))$.
	Now, using this isomorphism and the Riemann--Roch Theorem \ref{ELEM} applied to $\Lambda_{\DD}(A)$, we get $i(A)=\operatorname{dim}\left(\frac{\mathcal{A}_{\DD} }{\mathcal{A}_{\DD}(A)+\DD}\right)$.
	Thus, we have
        \begin{equation}
        \label{eq:speciality1}
	\operatorname{dim}\left(\frac{\mathcal{A}_{\DD} }{\mathcal{A}_{\DD}(A,W)+\DD}\right)= i(A)+ \operatorname{dim}\left(\frac{\mathcal{A}_{\DD}\left(A\right)+\DD}{\mathcal{A}_{\DD}(A,W)+\DD}\right).
        \end{equation}
        In order to relate the dimension of the latter quotient to $i(A, W)$, we consider the following commutative diagram with exact rows:
	$$\begin{aligned}&\begin{array}{cccccccccc}
        0 & \longrightarrow & \Lambda_{\DD}(A,W) & \longrightarrow & \mathcal{A}_{\DD}(A,W) & \longrightarrow & \frac{\mathcal{A}_{\DD}(A,W) + \DD}{\DD} & \longrightarrow & 0 \\
        & & \hookdownarrow & & \hookdownarrow & & \hookdownarrow & & \\
        0 & \longrightarrow & \Lambda_{\DD}(A) & \longrightarrow & \mathcal{A}_{\DD}(A) & \longrightarrow & \frac{\mathcal{A}_{\DD}(A) + \DD}{\DD} & \longrightarrow & 0\end{array}\end{aligned}.$$
	Applying the snake lemma to it, we get the following exact sequence of cokernels
	$$\begin{aligned}\label{EQ3_}
			0 \longrightarrow \frac{\Lambda_{\DD}\left(A\right) }{ \Lambda_{\DD}\left(A,W\right)} \longrightarrow \frac{\mathcal{A}_{\DD}\left(A\right) }{ \mathcal{A}_{\DD}\left(A,W\right) }\longrightarrow \frac{\mathcal{A}_{\DD}\left(A\right)+\DD}{\mathcal{A}_{\DD}\left(A,W\right)+\DD} \longrightarrow 0.
	\end{aligned}$$
        Hence
        \begin{equation}
        \label{eq:speciality2}
	\operatorname{dim}\left(\frac{\mathcal{A}_{\DD}(A)+\DD}{\mathcal{A}_{\DD}(A,W)+\DD}\right)=
          \operatorname{dim}\left(\frac{\mathcal{A}_{\DD}(A)}{\mathcal{A}_{\DD}(A,W)}\right) - \lambda\left(A\right)+\lambda\left(A,W\right).
        \end{equation}
        Finally, we observe that, for any $i\in[s]$, we have a second commutative diagram with exact rows:
	$$\begin{aligned}&\begin{array}{cccccccccc} 
        0 & \longrightarrow &\mathcal{A}_{\DD}\left(A,V_{i}\right)_{P_i}  &\longrightarrow &\mathcal{A}_{\DD}\left(A\right)_{P_i}  & \stackrel{\bar\varepsilon_i}{\longrightarrow} & \operatorname{Hom}_{\FF}(V_{i},V_{i}) & \longrightarrow & 0 \\
        & & \shortparallel & & \hookuparrow & & \hookuparrow & & \\
        0 & \longrightarrow & \mathcal{A}_{\DD}\left(A,V_{i}\right)_{P_i}  & \longrightarrow  & \mathcal{A}_{\DD}\left(A,W_{i}\right)_{P_i} & \stackrel{\bar\varepsilon_i}{\longrightarrow}  & \operatorname{Hom}_{\FF}(V_{i}/W_{i},V_{i})& \longrightarrow & 0
        \end{array}\end{aligned}$$
	from which we derive that $\mathcal{A}_{\DD}\left(A\right)_{P_i}  / \mathcal{A}_{\DD}\left(A,W_i\right)_{P_i}  \simeq \operatorname{Hom}_{\FF}(W_{i},V_{i})$.
	Therefore $\mathcal{A}_{\DD}(A) / \mathcal{A}_{\DD}(A,W)$ is isomorphic to $\bigoplus_{i \in [s]} \operatorname{Hom}_{\FF}(W_{i},V_{i})$ and thus has dimension $r \dim(W)$.
        Combining with Equations~\eqref{eq:speciality1} and~\eqref{eq:speciality2}, we get the desired result.
\end{proof}

\begin{Theorem}\label{RTHEO}
	For any extended divisor $(A,W)$, we have
	$$
		\lambda_{\DD}(A,W) = r\operatorname{deg}(A,W)+r(1-g)+\operatorname{dim}\left(\mathcal{A}_{\DD} \slash \left(\mathcal{A}_{\DD}(A,W)+\DD\right)\right).$$
	Moreover, if $\operatorname{deg}(A,W)<0$, then $\lambda_{\DD}(A,W)=0$.
\end{Theorem}

\begin{proof}
	The first equality results from Proposition \ref{ROCH}.
	The second assertion is a consequence of \cite[Theorem~3.5]{1}.
\end{proof}

\subsection{Differentials and residues}

In order to prove the Riemann--Roch Theorem for extended divisors, we now need to relate the quotient $\mathcal{A}_{\DD} \slash (\mathcal{A}_{\DD}(A,W)+\DD)$ to an actual Riemann--Roch space.
As in the classical case, we will use as an intermediate a certain space of differential forms over $\DD$.

\subsubsection{Definitions}

In the sequel, we denote by $\Omega_{\KK}$ (resp.~$\Omega_{\LL}$) the module of differentials of $\KK$ (resp.~$\LL$) as defined in \cite[Chapter 5]{8}.
For any differential $\eta \in \Omega_{\KK}$, we denote by $(\eta)_\KK$, (resp.~$(\eta)_\LL$) the divisor of $\eta$ viewed as a differential of $\Omega_{\KK}$ (resp.~$\Omega_{\LL}$).
They satisfy the Riemann--Hurwitz relation $(\eta)_\LL=\operatorname{CoNr}((\eta)_\KK)+\operatorname{Diff}(\LL/\KK)$ as stated in \cite[Theorem 3.4.6]{2}, where $\operatorname{Diff}(\LL/\KK)$ is the different divisor of $\LL/\KK$ as defined in \cite[Definition 3.4.3]{2}.

\begin{Definition}
	The \emph{left module of differentials} of $\DD$ is defined by $\Omega_{\DD}:=\DD \otimes_{\KK} \Omega_{\KK}.$
\end{Definition}

We recall that $\Omega_{\KK}$ is a one-dimensional vector space over $\KK$ \cite[Proposition 1.5.9]{2}. Henceforth, any element $\nu \in \Omega_{\DD}$ is of the form $\nu = f \otimes \eta$ for  $f\in \DD$ and $\eta \in \Omega_{\KK}$.

\begin{Definition}
Let $\nu = f \otimes \eta \in \Omega_\DD$ with $f \in \DD$ and $\eta \in \Omega_\KK$.
The \emph{principal divisor} of $\nu$ is
$$(\nu)_\DD := (f) + (\eta)_{\LL} + \Delta.$$
\end{Definition}

\begin{Proposition}
	The divisor $(\nu)_\DD$ does not depend on the choice of the decomposition.
\end{Proposition}
\begin{proof}
	Let  $\nu=f'\otimes {\eta'}$ be another decomposition. Then $\eta'=u \eta$ and $f = u f'$ for some nonzero $u \in \KK$.
	As for any $v\in \KK$, $g\in \DD$, and $Q\in \Pl_{\LL}$, we have $\omega_Q(vg)=\omega_Q(gv)=w_Q(v)+\omega_Q(g)$, we obtain
	$(f)+(\eta)_{\LL}= (f')+(u)+(\eta)_{\LL} = (f')+(\eta')_{\LL}$.
\end{proof}
\begin{Definition}\label{DEFOMEGA}
	For any $A \in \operatorname{Div}(\DD)$, we set $\Omega_{\DD}(A):=\big\{\nu \in \Omega_{\DD} \mid (\nu)_\DD   \geq A \big\}$.
\end{Definition}
\begin{Definition}\label{DEFRES}
Let $i \in [s]$ and let $t_i \in \KK$ be a uniformizer at $P_i$.
Let $\nu \in \Omega_\DD$ and write $\nu = f \frac{{d}t_i}{t_i}$ with $f \in \DD$.
If $\omega_{P_i}(f) \geq 0$, we define the \emph{residue endomorphism} of $\nu$ at $P_i$ by
$$\Res_{P_i}(\nu ) := \bar\varepsilon_i(f)$$
where $\bar\varepsilon_i$ is defined in Equation~\eqref{EPS}.
\end{Definition}
\begin{Proposition}\label{WELLDEF}
	The residue endomorphism $\Res_P(\nu)$ does not depend on the choice of the decomposition of $\nu$.
\end{Proposition}
\begin{proof}
Let $t'_P$ be another uniformizer of $\KK$ at $P$ and write $\nu = f'\frac{{d}t'_P}{{t'}_P}$ with $f' \in \DD$.
We then have $f = f' \frac{{d}t'_P}{{d}t_P}\frac{t_P}{{t'}_P}$ and so $\omega_P(f')=\omega_P(f)$.
Moreover $\varepsilon_P(f) = \varepsilon_P\left(f' \frac{{d}t'_P}{{d}t_P}\frac{t_P}{{t'}_P}\right) =\varepsilon_P(f') $.
Indeed, writing $t'_P=\tau_0 t_P (1+t_P(\tau_1+t_P \rho))$ with $\rho \in \mathcal{O}_P$, $\tau_1 \in \FF$ and $\tau_0 \in \FF^\times$, we compute ${\frac{{d}t'_P}{{d}t_P}}\equiv {\frac{t'_P}{t_P}}\equiv\tau_0 \bmod P.$
Finally, by definition, $v_P(t'_P)=v_P(t_P)=1$. This ensures that $\tau_0 \neq 0 \bmod P$.
\end{proof}

\subsubsection{The residue formula}

We observe that the residue endomorphisms at different places do not live in the same ambient space.
Therefore, adding them does not make sense and we then cannot hope that their sum vanishes.
However, this vanishing property does hold after taking traces, as we show now.

We first recall the notion of reduced trace on $\DD$.

\begin{Definition}
The \emph{reduced trace} of an element $f$ of $\DD$ (resp.~of $\DD_P$), denoted by $\Trd(f)$,
is the trace of left-$\LL$ (resp.~left-$\LL_P$) linear endomorphism of $\DD$ (resp.~$\DD_P$) given by right multiplication by $f$.
\end{Definition}

\begin{Lemma}\label{PRELLEMMA}
	Let $P \in \Pl_\LL$ and let $f =f_0 +f_1 T +\cdots + f_{r-1}T^{r-1}\in \DD_P $ with $f_j \in \LL_P$.
	Then $\Trd(f)=\operatorname{Tr}_{\LL_P/\KK_P}(f_0)$.

	If moreover $P = P_i$ for some $i \in [s]$, then $\Trd(f) = \operatorname{Tr}_{\KK_{P_i}}\left(\varepsilon_i(f)\right)$
	(where we recall that $\varepsilon_i$ is the morphism defined at the beginning of Subsection~\ref{RAPPEL}).
\end{Lemma}
\begin{proof}
	It is similar to the proof of \cite[Lemma 2]{1}.
\end{proof}

\begin{Lemma}\label{TREDSUP}
	Let $P \in \Pl_{\KK}$, and $f \in \DD_P$.
	We have $v_P(\Trd(f))\geq \omega_P(f).$
\end{Lemma}
\begin{proof}
	Let $f=f_0+f_1 T+\cdots+f_{r-1} T^{r-1} \in \DD_P$ with all $f_i \in \LL_P$.
	By Lemma \ref{PRELLEMMA} we have
	$\Trd(f)=\operatorname{Tr}_{\LL/\KK}(f_0).$
	We conclude by the ultrametric inequality of $\omega_P$ that
	$v_{P}\left(\Trd(f)\right)=\omega_P\left(\operatorname{Tr}_{\LL/\KK}(f_0)\right)\geq \operatorname{min}_{0 \leq i < r}(\omega_P(\theta^{i}(f_0))) = \omega_P(f_0) \geq \omega_P(f)$.
\end{proof}

\begin{Definition}\label{REDTrOMEGA}
	We define the \emph{reduced trace} on $\Omega_{\DD}$ by
	$$\begin{array}{rcl}
	\Trd: \quad \Omega_{\DD} & \longrightarrow & \Omega_{\KK} \\
	f \otimes \eta & \mapsto & \Trd(f) \cdot \eta
	\end{array}$$
	For $\nu \in \Omega_{\DD}$,
	we define the \emph{reduced residue} of $\nu$ at a place $P$ by
	$$\Resrd_{P}(\nu ) :=\res_{P}(\Trd(\nu)),$$
	where the right-hand side is the classical residue of $\Trd(\nu) \in \Omega_{\KK}$ at the place $P$.
\end{Definition}

\begin{Proposition}\label{CRUCLEMMA}
	Let $i \in [s]$. Then,
	for any $\nu \in \Omega_\DD$ such that $\Res_{P_i}(\nu)$ is defined, we have
	$\Resrd_{P_i}(\nu) =\operatorname{Tr}_{\FF}(\Res_{P_i}(\nu ) )$.
\end{Proposition}
\begin{proof}
	We write $\nu = f \frac{d t_i}{t_i}$ where $t_i$ is a uniformizer at $P_i$, so that $\omega_{P}(f)\geq 0$ and $\Res_{P_i}(\nu ) = \bar\varepsilon_i(f)$.
	From Lemma~\ref{TREDSUP}, we deduce $v_P(\Trd(f))\geq 0$ and then
	$$\textstyle \Resrd_{P_i}(\nu) = \res_{P_i}\big(\Trd(f)\frac{{d}t_i}{t_i}\big)= \Trd(f) \bmod P_i.$$
	On the other hand, Lemma~\ref{PRELLEMMA} ensures that $\Trd(f) = \operatorname{Tr}_{\KK_{P_i}}(\varepsilon_i(f))$.
	Reducing modulo $P_i$, we obtain $\Trd(f) \bmod P_i = \operatorname{Tr}_{\FF}(\bar\varepsilon_i(f)) = \operatorname{Tr}_{\FF}(\Res_{P_i}(\nu ) )$.
\end{proof}

\begin{Theorem}[Residue Theorem] \label{ResTHEO}
	For $\nu \in \Omega_{\DD}$, we have $$\sum_{P \in \Pl_{\KK}} \Resrd_{P}(\nu ) = 0.$$
\end{Theorem}
\begin{proof}
	 We apply the residue theorem as stated in \cite[Theorem 3.27]{8} to $\Trd(\nu) \in \Omega_{\KK}$ and conclude.
\end{proof}

\subsubsection{Riemann--Roch spaces}\label{RROmega}

For all $i \in [s]$, we endow the space $V_i$ with the bilinear form 
$$\begin{array}{rcl}
\langle-,-\rangle: \quad V_i \times V_i & \rightarrow & \FF \\
(v_i, w_i) &\mapsto & \operatorname{Tr}_{{V_{i}}/{\FF}}({v_i} {w_i}).
\end{array}$$
It is nondegenerate given that the place $P_i$ is supposed to be unramified in the extension $\LL/\KK$.
If $W_i$ is a vector subspace of $V_i$, we write $W_i^\perp$ for its orthogonal and,
similarly, if $W = \bigoplus_{i=1}^s W_i$ is a subspace of $V_D$, we define $W^\perp = \bigoplus_{i=1}^s W_i^\perp$.

We now introduce Riemann--Roch spaces of differential forms: if $(A,W)$ is an extended divisor, we set
$$\Omega_{\DD}(A,W):=\big\{\nu \in \Omega_{\DD}(A-\operatorname{CoNr}(D)) \mid  \Res_{P_i}(\nu)_{|W_{i}^\perp}=0, \,\forall  i \in[s]\big\}.$$
It follows from Lemma \ref{VANISHING} that $\Omega_{\DD}(A,0)=\Omega_{\DD}(A)$ for any  $A\in \operatorname{Div}(\DD,D)$.

It turns out that Riemann--Roch spaces of differential forms are related to classical Riemann--Roch spaces.
In order to make this relationship explicit, we choose $t_D \in \KK$ such that $v_{P_i}(t_D) = 1$ for all $i \in [s]$ and set
\begin{equation}\label{eq:defcan}
\textstyle \can := \big(\frac{d t_D}{t_D}\big)_\DD = \big(\frac{d t_D}{t_D}\big)_\LL + \Delta.
\end{equation}
It is a canonical divisor of $\DD$ in the sense of Definition~\ref{DEFK}.

\begin{Definition}\label{EXTDUAL}
For an extended divisor $(A,W)$, we set $(A,W)^\star:=({\can+\operatorname{CoNr}(D)-A},W^\perp)$.
\end{Definition}

\begin{Theorem}\label{ISOGA}
	For any extended divisor $(A,W)$, we have the isomorphism:
	$$\begin{array}{rcl}
          {\Lambda_{\DD}((A,W)^\star)} &\stackrel{\sim}{\longrightarrow}& {\Omega_{\DD}(A,W)} \\
		f &\mapsto & f\frac{dt_D}{t_D}
	\end{array}$$
\end{Theorem}

\begin{proof}
It follows directly from the definitions.
\end{proof}

\subsection{Serre duality and Riemann--Roch Theorem} \label{SERREDUALITY}

The last step of the proof of the Riemann--Roch Theorem consists in relating the adelic quotient $\mathcal{A}_{\DD} \slash (\mathcal{A}_{\DD}(A,W)+\DD)$
to the space of differential forms $\Omega_{\DD^\star}(A, W)$, where $\DD^\star$ is the adjoint algebra as introduced in Definition~\ref{DEFADJOINTALGEBRA}.

We recall from Definition~\ref{DEFADJUNCTION} that we have already introduced an adjunction $\DD \to \DD^\star$, $f \mapsto f^\star$ at the level of Ore polynomials.
It naturally extends to differential forms as follows:
$$\begin{array}{rcll}
	\Omega_{\DD} & \longrightarrow & \Omega_{{\DD}^\star} \\
	\nu = f \otimes \eta       & \mapsto & {\nu}^{\star}:={f}^{\star}\otimes \eta
\end{array}$$
with, as usual, $f \in \DD$ and $\eta \in \Omega_\KK$.
Besides, for all $i \in [s]$, the algebra $\operatorname{End}_{\FF}(V_{i})$ is also endowed with an adjunction $\varphi_i \mapsto \varphi_i^\star$ defined by the standard property
$$\langle \varphi_i(v_i),w_i\rangle = \langle v_i, \varphi_i^{\star}(w_i)\rangle \qquad \text{(}v_i, w_i \in V_i\text{)}.$$
We recall the classical formulas $\ker(\varphi_i^\star) = \operatorname{im}(\varphi_i)^\perp$ and $\operatorname{im}(\varphi_i^\star) = \ker(\varphi_i)^\perp$.

If $\varphi$ is a tuple $(\varphi_1, \ldots, \varphi_s)$, we will often write $\varphi^\star = (\varphi_1^\star, \ldots, \varphi_s^\star)$ for simplicity.

\begin{Lemma} \label{COMMLEMMA}
	Let $A \in \operatorname{Div}(\DD,D)$ and let $f \in \DD$.
	If $i \in [s]$ is such that $\omega_{P_i}(f) \geq 0$, then
	$$\textstyle
	\bar\varepsilon_i(f)^\star ={{\bar\varepsilon}_i({f}^{\star})}
	\quad \text{and} \quad
	\Resrd_{P_i}\big(f \frac{dt_D}{t_D}\big)^\star =
	\Resrd_{P_i}\big(f^\star \frac{dt_D}{t_D}\big)$$
	where, in a slight abuse of notation, we continue to write $\bar\varepsilon_i$ for the evaluation morphism on $\DD^\star$.
\end{Lemma}
\begin{proof}
	We reuse the notation $u_{P_i}$ introduced in Equation \eqref{H2PREC}.
	By linearity, it is enough to check that
	$\langle{\bar\varepsilon}_i({f}^{\star})(v_i),w_i\rangle=\langle v_i,\bar\varepsilon_i(f)(w_i)\rangle$ for $f=f_j T^j$ and $v_i, w_i \in V_i$.
	Setting $\operatorname{Nr}_{a,b}(u_P) := \prod_{k=a}^{b-1} \theta^k(u_P)$, we compute
	$$\bar\varepsilon_i({f}^{\star})=  \bar\varepsilon_i\left(\frac{T^{r-j}}{x^{-1}} f_j\right)=\frac{\operatorname{Nr}_{0,r-j}(u_P^{-1})\theta^{r-j} (f_j)}{\operatorname{Nr}_{0,r}(u_P^{-1})}\theta^{r-j} =\frac{\theta^{r-j} (f_j)}{\operatorname{Nr}_{r-j,r}(u_P^{-1})}\theta^{r-j}$$
	as changing $x \mapsto x^{-1}$ in ${\DD}^\star$ induces the change $u_{P_i} \mapsto u_{P_i}^{-1}$.
	Using  $\theta^j(\operatorname{Nr}_{r-j,r}(u_{P_i}))=\operatorname{Nr}_{0,j}(u_{P_i})$, we get
        \begin{align*}
	\operatorname{Tr}_{{V_{i}}/{\FF}}\left(\frac{\theta^{r-j} (f_j)}{\operatorname{Nr}_{r-j,r}(u_{P_i}^{-1})}\theta^{r-j}({v_i}) {w_i}\right)
        & = \sum_{k=0}^{r-1} \theta^k\left( \frac{\theta^{r-j} (f_j)}{\operatorname{Nr}_{r-j,r}(u_{P_i}^{-1})}\theta^{r-j}({v_i}) {w_i}\right) \\ 
        & = \operatorname{Tr}_{{V_{i}}/{\FF}}({v_i}  f_j \operatorname{Nr}_{0,j}(u_{P_i})\theta^{j}({w_i})),
        \end{align*}
	which concludes the proof.
\end{proof}

\begin{Theorem} \label{DUALITY}
	For any extended divisor $(A,W)$, the $\FF$-bilinear form
	$$\begin{array}{rcl} \Phi : \,
        \mathcal{A}_{\DD} / (\mathcal{A}_{\DD}(A,W)+\DD) \times  \Omega_{{\DD}^\star}(A,W) & \longrightarrow & \FF \\
	(\alpha, \nu) & \mapsto & \sum\limits_{P \in \Pl_{\KK}} \Resrd\left(\alpha_{P} {\nu}^{\star} \right)
        \end{array}$$
	is a perfect pairing.
\end{Theorem}
	Before diving into the proof, we show that this pairing is well-defined.
	\begin{Lemma}\label{WELLDEF2}
		For ${\nu} \in \Omega_{{\DD}^\star}(A,W)$,
		$\alpha\in \mathcal{A}_{\DD}(A,W)$ and $a\in \DD$,
		we have $\Phi(\alpha+a,\nu) =0$.
	\end{Lemma}
	\begin{proof}
		As we have $a{\nu}^{\star} \in \Omega_{{\DD}}$,  we get
		$\Phi(a,\nu) =0$ by applying Theorem \ref{ResTHEO} to $a{\nu}^{\star}$.

		It remains to prove $\Phi(\alpha,\nu) =0$.
		Let $P \in \Pl_\KK$.
		We first assume that $P$ is one of the $P_i$, $i \in [s]$. Then $\Res_{P_i}(\alpha_P \nu^\star)$ is defined and
		$$\Res_{P_i}(\alpha_P \nu^\star) = \bar\varepsilon_i(\alpha_P) \circ \Res_{P_i}(\nu^\star) = \bar\varepsilon_i(\alpha_P) \circ \Res_{P_i}(\nu)^\star.$$
		Moreover, we know that $\bar\varepsilon_i(\alpha_P)$ and $\Res_{P_i}(\nu)$ vanish on $W_i$ and $W_i^\perp$ respectively. In other words
		$\operatorname{im} \Res_{P_i}(\nu)^\star = \big(\!\ker \Res_{P_i}(\nu)\big)^\perp \subset W_i \subset \ker \bar\varepsilon_i(\alpha_P)$.
		We deduce that $\Res_{P_i}(\alpha_P \nu^\star)$ vanishes and so does its trace, \ie $\Resrd_{P_i}(\alpha_P \nu^\star) = 0$.

		We now consider the case where $P$ is not an evaluation place.
		We write $\nu= f \otimes \eta$ where $f \in \DD$ and $\eta \in \Omega_\KK$ is chosen such that it has neither zero nor pole at $P$.
		We write $\alpha_P=\sum_{j=0}^{r-1}{\alpha_j} T^j$ (with $\alpha_j \in \LL_P$)
		and $f = \sum_{j=0}^{r-1}f_j T^j$ (with $f_j \in \LL$).
		It follows from Lemma~\ref{PRELLEMMA} that
		$$\Resrd_P\left(\alpha_{P} {\nu}^{\star} \right) 
		= \res_P\big(\Trd(\alpha_P f^\star) \cdot \eta\big)
                = \sum_{j=0}^{r-1} \res_P\big(\operatorname{Tr}_{\LL_P/\KK_P}(\alpha_j f_j) \cdot \eta\big).$$
		On the other hand, we deduce from the definition of the Riemann--Roch spaces that, for any place $Q\mid P$, we have the estimations
		\begin{align*}
		w_Q(a_j) & \textstyle \geq -A_Q - \frac j r v_P(x) \\
		w_Q(f_j) + w_Q(\operatorname{Diff}(\LL/\KK)) & \textstyle \geq A_Q + \frac j r v_P(x) - \frac{b_P - 1}{b_P}
		\end{align*}
		where $\operatorname{Diff}(\LL/\KK)$ denotes the different of $\LL/\KK$. Adding the two inequalities and replacing the right hand sides by their ceilings, we find that
		$w_Q(a_j f_j) \geq -w_Q(\operatorname{Diff}(\LL/\KK))$, as $\lceil -\frac{b_P - 1}{b_P} \rceil \geq 0$. 
		Hence, by definition of the different, $\operatorname{Tr}_{\LL_P/\KK_P}(a_j f_j)$ has no pole at $P$.
		Therefore the residue at $P$ of $\operatorname{Tr}_{\LL_P/\KK_P}(a_j f_j) \cdot \eta$ vanishes, and we conclude that
		$\Resrd_P\left(\alpha_{P} {\nu}^{\star} \right) = 0$.
	\end{proof}
	\begin{proof}[Proof of Theorem \ref{DUALITY}]
	By Lemma \ref{WELLDEF2} the pairing $\Phi$ is well-defined. The proof of its non-degeneracy breaks down into two parts, proving the injectivity and then the surjectivity of the morphism $$\begin{array}{rcl} \psi: \quad
			\Omega_{{\DD}^\star}(A,W) & \rightarrow & \operatorname{Hom}\left(\mathcal{A}_{\DD} /(\mathcal{A}_{\DD}(A,W)+\DD),\FF\right) \\
			\nu                           & \mapsto & \big(\alpha \mapsto \Phi(\alpha, \nu)\big)
		\end{array}.$$

		\noindent \emph{Injectivity.}
		      Let $P$ be a fixed evaluation place.
		      Since $P$ is unramified, we can find an element $a \in \mathcal O_{\LL,P}$ such that $\operatorname{Tr}_{\mathcal{O}_{\LL,P}/\mathcal{O}_P}(a) \equiv 1 \pmod P$.
		      We consider the adele $\beta \in \mathcal A_\DD$ with coordinate $a$ at $P$ and $0$ elsewhere.
                      Let $\nu = f \otimes \frac{{d}t_P}{t_P} \in \Omega_{\DD^\star}(A,W)$, $\nu \neq 0$.
                      Since $\DD$ is a division algebra, $f$ is invertible and we can form the adele $\alpha = \beta \cdot (f^{-1})^{\star} \in \mathcal A_\DD$.
		      Then $$\textstyle \Phi(\alpha,\nu) = \Resrd_P\big(a \frac{{d}t_P}{t_P}\big) = \operatorname{Tr}_{\LL_P/\KK_P}(a) \bmod P = 1$$
		      and injectivity follows.

		\medskip

		\noindent \emph{Surjectivity.}
		      By injectivity of $\psi$ and Theorem \ref{ISOGA}, we have $$\lambda_{\DD^\star}((A,W)^\star)\leq i(A,W)=\lambda_{\DD}(A,W)-(r\operatorname{deg}(A,W)+r(1-g)).$$
		      Besides, since ${\can+\operatorname{CoNr}(D)-A}  \in \operatorname{Div}(\DD^\star,D)$,
		      we can apply the previous inequality to the extended divisor $(A,W)^\star$.
		      It yields $$\begin{aligned}\lambda_{\DD}(A,W) & \leq \lambda_{\DD^\star}((A,W)^\star)-(r\operatorname{deg}((A,W)^\star)+r(1-g))\\
				& \leq \lambda_{\DD}(A,W)-r\operatorname{deg}(A,W)-r\operatorname{deg}((A,W)^\star)+r(2g-2)
				=\lambda_{\DD}(A,W)\end{aligned}$$
				as $\operatorname{deg}(A,W)+\operatorname{deg}((A,W)^\star)=2g -2 .$
		      Thus, all inequalities have to be equalities and we deduce $\lambda_{\DD^\star}((A,W)^\star)=i(A,W)$. We conclude by invoking Theorem \ref{ISOGA}.
\end{proof}
Finally, we can state and prove our Riemann--Roch Theorem for extended divisors.
\begin{Theorem}[Riemann--Roch Theorem]\label{RROCHTHEO}
	For any extended divisor $(A,W)$, we have
	$$\lambda_{\DD}(A,W)   = r\operatorname{deg}(A,W)+r(1-g)+\lambda_{\DD^\star}((A,W)^\star).$$
\end{Theorem}
\begin{proof}
	It results by combining Theorem \ref{RTHEO},  Theorem \ref{ISOGA} and Theorem \ref{DUALITY}.
\end{proof}

\section{Application to linearized Algebraic Geometry codes}\label{LAGCODEDUALITYANDDECODING}

In all what follows, we consider the same setting as in Section~\ref{RRTHEOSEC} and refer to the introduction of that section for the notation.
In particular, the field $\FF$ can still be arbitrary while, for applications to coding theory, the cases of most interest are finite fields $\FF=\mathbb{F}_q$.

To start with, we recall the construction of linearized AG codes as settled in~\cite{1}.
As in Subsection~\ref{RAPPEL}, we fix a finite set of unramified places $\{P_1, \ldots, P_s\} \subset \Pl_\KK$ satisfying the hypothesis \Htpl{P_i}
and we let $V_i$ denote the residue algebra of $\LL$ at $P_i$.
The ambient space in which our codes live is
$$\mathcal H := \End_\FF(V_1) \times \cdots \times \End_\FF(V_s).$$
By choosing $\FF$-bases of the $V_i$, $\mathcal H$ becomes isomorphic to the $s$-th power of a matrix algebra, namely $(\FF^{r \times r})^s$.
However, in the algebraic context of this article, it will be convenient to work with endomorphisms, given that the evaluation morphisms $\bar\varepsilon_i$ (see Equation~\eqref{EPS}) naturally produce endomorphisms.
In any case, $\mathcal H$ is a vector space of dimension $sr$ over $\FF$ and we endow it with the sum-rank metric defined by
$$\dsrk(\varphi, \psi) = \wsrk(\varphi - \psi)
\quad \text{with} \quad
\wsrk(\varphi_1, \ldots, \varphi_s) = \sum_{i=1}^s \operatorname{rank}(\varphi_i).$$
We set $D := P_1 + \cdots + P_s \in \operatorname{Div}(\KK)$.

\begin{Definition}
Let $A  \in \operatorname{Div}(\DD,D)$. The \emph{linearized Algebraic Geometry} code associated to $A$ and $D$ is the image of
$$\bar\varepsilon_D : \, \Lambda(A) \longrightarrow \mathcal H, \quad f \mapsto \big(\bar\varepsilon_1(f), \ldots, \bar\varepsilon_s(f)\big).$$
It is denoted by $\LAG_\DD(A, D)$.
\end{Definition}

The aim of this section is to apply the theory we developed previously to characterize the duals of linearized AG codes and to decode them.
In the special case where $\KK = \FF(t)$, $\LL = \FF'(t)$ for a finite cyclic extension $\FF'/\FF$ and $A$ is a multiple of the point at infinity,
the code $\LAG_\DD(A, D)$ is a linearized Reed--Solomon code~\cite{36}, for which the duals have already been described in~\cite{34} and a decoding algorithm is known~\cite{7}.
Another corner case corresponds to the situation where the extension $\LL/\KK$ is trivial: in this case, $\LAG_\DD(A, D)$ is nothing but a classical AG code, and our results extend the classical ones.

We also mention that, as it was pointed out in \cite[Section~IV]{31}, linearized Reed--Muller codes~\cite{31} can be embedded in larger linearized AG codes and hence can be decoded using our decoding algorithm.

\subsection{Duality}\label{LAGCODEDUALITY} 

We recall that each $V_i$ is endowed with a nondegenerate bilinear pairing (see Subsection~\ref{RROmega}), which allows to define an adjunction $\varphi_i \mapsto \varphi_i^\star$ on $\End_\FF(V_i)$.
We use it to endow $\mathcal H$ with another nondegenerate bilinear pairing defined by:
$$\begin{array}{rcl}
\mathcal H \times \mathcal H & \longrightarrow & \FF \\
(\varphi, \psi) & \mapsto & \langle \varphi, \psi \rangle :=\sum_{i=1}^{s}\operatorname{Tr}_\FF\left(\varphi_i \circ \psi^{\star}_i\right)
\end{array},$$
where $\varphi_i$ and $\psi_i$ are the coordinates of $\varphi$ and $\psi$ respectively.
We define the dual of a linearized AG code with respect to this pairing, as follows:
$$\operatorname{LAG}_{\DD}(A,D)^\perp =\big\{\psi \in \mathcal H \,\big|\, \langle \varphi, \psi \rangle = 0, \, \forall \varphi \in \operatorname{LAG}_{\DD}(A,D)\big\}.$$
We will now apply the duality theory developed in Subsection~\ref{SERREDUALITY} to define linearized Differential (LD) codes and subsequently show that they are the duals of
linearized AG codes.

\begin{Definition}\label{CODEDIFF}
	The \emph{linearized Differential code} associated to the divisors $A$ and $D$, denoted by $\operatorname{LD}_{\DD}(A,D)$, is the image of the following map:
	$$\begin{array}{rcl}
		\Omega_{\DD^\star}(A,V_D) & \longrightarrow & \mathcal H \\
		\nu & \mapsto & \big(\Res_{P_1}(\nu), \ldots, \Res_{P_s}(\nu)\big)
	\end{array}$$
        where $\Res$ is the endomorphism residue introduced in Definition~\ref{DEFRES}.

\end{Definition}
\begin{Theorem}\label{NONCAN}
	Let $\can$ be the canonical divisor defined in Equation~\eqref{eq:defcan}. Then, we have the following equalities of codes:
	$$\LAG_{\DD}(A,D)^\perp = \LD_{\DD^\star}(A,D) = \LAG_{\DD^\star}({\can+\operatorname{CoNr}(D)-A},D).$$
\end{Theorem}
\begin{proof}
        We have the following commutative diagram
	$$\begin{tikzcd}[column sep=5em, row sep=2em]
		{\Lambda_{\DD^\star}({\can+\operatorname{CoNr}(D)-A})} \arrow[rightarrow]{d}{(\bar\varepsilon_i)_{i}} \arrow[r, "f \mapsto f \frac{{d}t_D}{t_D}" ] 
		& {\Omega_{\DD^\star}(A,V_D)} \arrow[rightarrow]{d}{(\Res_{P_i})_{i}}  \\
		\mathcal H \arrow[r, "\operatorname{Id}"] & \mathcal H
	\end{tikzcd}$$
        whose top arrow is an isomorphism thanks to Theorem~\ref{ISOGA}.
	The equality between $\LD_{\DD^\star}(A,D)$ and $\LAG_{\DD^\star}(\can+\operatorname{CoNr}(D)-A,\,D)$ follows.

	Let $f \in \Lambda_{\DD}(A)$ and $\nu \in \Omega_{\DD^\star}(A,V_D)$.
        By the Residues Theorem (Theorem~\ref{ResTHEO}), we know that $\sum_{P \in \Pl_\KK} \Resrd_P(f \nu^\star) = 0$.
	Moreover, it follows from the proof of Lemma~\ref{WELLDEF2} that
	$$\Resrd(f \nu^\star) = \begin{cases}
	\operatorname{Tr}_\FF\big(\bar\varepsilon_i(f) \circ \Res_{P_i}(\nu)^\star\big) & \text{if } P = P_i, \\
	0 & \text{if } P \not\in \{P_1, \ldots, P_s\}.
	\end{cases}$$
	We deduce that the codewords attached to $f$, on the one hand, and $\nu$, on the other hand, are orthogonal in $\mathcal H$, \ie
	$\LAG_{\DD}(A,D)\subset \LD_{\DD^\star}(A,D)^\perp$.
	To show the reverse inclusion, we compare the dimensions. First of all, noticing that the kernel of $\bar\varepsilon_D$ is $\Lambda(A, V_D)$, we find
	$$\dim_\FF \LAG_\DD(A, D) = \lambda_\DD(A) - \lambda_\DD(A, V_D).$$
	Then, using the first part of the theorem, we obtain similarly
	\begin{align*}
	\dim_\FF \LD_{\DD^\star}(A, D) & = \lambda_{\DD^\star}(\can + \operatorname{CoNr}(D)-A) - \lambda_{\DD^\star}(\can + \operatorname{CoNr}(D)-A, V_D) \\
	& = sr^2 - \lambda_\DD(A) - \lambda_\DD(A, V_D)
	\end{align*}
	the last equality coming from the Riemann--Roch Theorem (Theorem~\ref{RROCHTHEO}) and the fact that $\dim_\FF V_D = sr$.
	Observing that the two dimensions sum to $sr^2 = \dim_\FF \mathcal H$ concludes the proof.
\end{proof}
\begin{Corollary}
	We assume $2g-2<\operatorname{deg}A <sr$.
	Then the code $\LD_{\DD^\star}(A, D)$ has length $sr^2$, dimension $sr^2 - r\deg A + r(g-1)$ over $\FF$ and minimum distance at least $\deg A - (2g - 2)$.
\end{Corollary}
\begin{proof}
	This results from the estimates on the parameters of a linearized AG code given in \cite[Theorem~2]{1}.
\end{proof}

\begin{Remark}
As recalled in \cite{13}, a maximum sum-rank distance (MSRD) code in the sum-rank metric is a code achieving the sum-rank metric Singleton-type bound $rd + k = n + r$ where $[n,k,d]$ are the code parameters: length, dimension, and minimum distance.
In our situation, both $\LAG_\DD(A, D)$ and its dual $\LD_{\DD^\star}(A, D)$ meet the Singleton bound up to a defect of $rg$, \ie they both satisfy $rd + k\geq n + r(1-g)$.
\end{Remark}

\subsection{Decoding}\label{DECODINGLAG}

In what follows we describe and prove a decoding algorithm for linearized AG codes with unramified evaluation places. 
Our algorithm can be seen as the sum-rank metric version of the one given in \cite[Section 8.5]{2} for AG codes in the Hamming metric. 
We make the usual injectivity hypothesis on the evaluation maps giving the code and its dual, respectively, as follows.
\begin{Hypothesis}\label{HYPO}
	From now on, we assume the following condition on $A$:
	$$2g -2   < \operatorname{deg}A<sr .$$
\end{Hypothesis}
\begin{Remark}\label{GALINTEGER}
	The map $\bar\varepsilon_D$ defining the code $\LAG_\DD(A, D)$ is injective if and only if $\lambda_{\DD}(A,V_D)=0$.
	This latter condition is fulfilled if $\operatorname{deg}A < sr$.
	Similarly, the defining morphism of the dual code $\operatorname{LAG}_{\DD^\star}({\can+\operatorname{CoNr}(D)-A},D)$ is injective when $\operatorname{deg}A > 2g -2$.
	The two inequalities of Hypothesis \ref{HYPO} thus imply the injectivity of both $\operatorname{LAG}_{\DD}(A,D)$ and its dual.
\end{Remark}

We shall need the following consequence of the Riemann--Roch Theorem.

\begin{Theorem}\label{FUNDTHEO}
	Let $(A,W)$ be an extended divisor satisfying $2g -2<\operatorname{deg}(A,W)$.
	Then, the following short sequence is exact
	$$0 \rightarrow {\Lambda}(A,W)\rightarrow {\Lambda}(A)	\xrightarrow{\bar\varepsilon_{D|W}} \prod_{i=1}^{s}\operatorname{Hom}_{\FF}(W_{i},V_{i})   \rightarrow 0$$
	where, by definition, $\bar\varepsilon_{D|W}(f) = \big(\bar\varepsilon_1(f)_{|W_1}, \ldots, \bar\varepsilon_s(f)_{|W_s}\big)$.
\end{Theorem}
\begin{proof}
	The short sequence is left exact by definition. 
	By applying Theorem~\ref{RROCHTHEO} to ${\Lambda}(A)$ and ${\Lambda}(A,W)$, we check that the difference of their dimensions is $r\operatorname{dim}(W)$. We then conclude to the surjectivity of $\bar\varepsilon_{D|W}$.
\end{proof}

We will also need a variant of the pairing $\langle -, -\rangle$ relative to a subspace $W = \bigoplus_{i=1}^s W_i$ with $W_i \subset V_i$.
Precisely, given such a subspace $W$, we consider the $\FF$-bilinear form
$$\begin{array}{rcl} \Phi_W : \,
\prod_{i=1}^{s}\operatorname{Hom}_{\FF}(V_{i},W_{i}) \times \prod_{i=1}^{s}\operatorname{Hom}_{\FF}(W_{i},V_{i}) & \rightarrow & \FF \\
(\varphi, \psi) & \mapsto & \sum_{i=1}^{s}\operatorname{Tr}\left( \varphi_i \circ \psi_i \right)
\end{array}$$
where $\varphi_i$ and $\psi_i$ stand for the coordinates of $\varphi$ and $\psi$ respectively.

\begin{Lemma}\label{PAIRING}
The bilinear form $\Phi_W$ is nondegenerate.
\end{Lemma}
\begin{proof}
If there existed a nonzero $s$-tuple of homomorphisms $(\varphi_1, \ldots, \varphi_s)$ with $\varphi_i \in \operatorname{Hom}_{\FF}(V_{i},W_{i})$, such that, say, $a_{i_0}\neq 0$ and
$\sum_{i=1}^{s}\operatorname{Tr}\left(a_i\operatorname{Hom}_{\FF}(W_{i},V_{i})\right) = 0$,
then taking $\psi_i =0 \in \operatorname{Hom}_{\FF}(W_{i},V_{i})$ for $i \neq i_0$, we would have
$\operatorname{Tr}\left(\varphi_{i_0}\operatorname{Hom}_{\FF}(W_{i_0},V_{i_0})\right) = 0$.
This contradicts \cite[Lemma 4.2]{34}.
\end{proof}

In all what follows, we set 
$$\rhoalgo := \left \lfloor \frac{\ddesign(A)-g-1}{2}\right \rfloor =
\left \lfloor \frac{sr-\operatorname{deg}A-g-1}{2}\right \rfloor.$$
It is the maximal number of errors that we intend to decode with our algorithm.

\subsubsection{An auxiliary divisor}

\begin{Definition}\label{GALDIV}
A divisor $A \in \operatorname{Div}(\DD)$ is called \emph{Galois} if $A_Q=A_{Q'}$ when $Q$ and $Q'$ are above the same place.
\end{Definition}

\begin{Lemma} \label{LAGBELG1}
	Let $A_1, A_2 \in \operatorname{Div}(\DD)$. If $A_2$ is Galois then
	$\Lambda_{\DD^\star}(A_1)\Lambda_{\DD^\star}(A_2) \subset \Lambda_{\DD^\star}(A_1 + A_2)$.
\end{Lemma}
\begin{proof}
	Let us denote by $Q_j$ the $m_P$ places above the place $P \in \Pl_{\KK}$.
	If we do not apply $\min_{Q_j |P}$ and thus replace $\omega_P$ by $\omega_{Q_j}$
	in the proof of the surmultiplicativity of the gauge of Lemma \ref{SURMULT}, we get for any $f$ in $\Lambda_{\DD^\star}(A_1)$ and for any $g$ in $\Lambda_{\DD^\star}(A_2)$ that
	$$\omega_{Q_j}(fg) \geq \min_{0\leq k \leq i_0}\left(w_{Q_j}(f_{k})+k\frac{ w_{Q}(x)}{r}+w_{Q_{(j+k)\%m_P}}(g_{i_0-k})+(i_0-k)\frac{  w_{Q}(x)}{r}\right).$$
	Hence, for some $ k_0 \in \{0,\ldots,i_0\}$, we have
	$$\omega_{Q_j}(fg) \geq -(A_1)_{Q_j}-(A_2)_{Q_{(j+k_0)\%m_P}} =  -(A_1)_{Q_j} -(A_2)_{Q_j}=-(A_1 + A_2)_{Q_j},$$
	where the first equality holds since $A_2$ is a Galois divisor.
\end{proof}

We make the following assumption on $\LL/\KK$, which will be helpful to construct Galois divisors.

\begin{Hypothesis}\label{HYPO2}
	There exists a place
	$P \in \Pl_{\KK}\backslash\{P_1,\ldots,P_s\}$  such that $e_P f_P=r$ and $\operatorname{deg}P=1$.
\end{Hypothesis}

\begin{Lemma}\label{PARAM}
	There exists a divisor $A_1 \in  \operatorname{Div}(\DD^\star,D)$ such that
	\begin{enumerate}[(i)]
	\item $A + A_1$ is Galois,
	\item $\deg(A + A_1) < sr - \rhoalgo$,
	\item $\lambda_\DD(A_1) > r \cdot \rhoalgo$.
	\end{enumerate}
\end{Lemma}
\begin{proof}
	Using Hypothesis~\ref{HYPO2}, we can pick a Galois divisor $G$ of degree $sr - \rhoalgo - 1$ with support concentrated at the place $P$.
	We set $A_1 := G - A$. The first and second conditions are then clearly fulfilled.
	For the last one, we use the Riemann--Roch Theorem (Theorem \ref{ELEM}), which implies
	$$\frac{\lambda_\DD(A_1)} r \geq \deg(A_1) - g + 1 = \deg(G) - \deg(A) - g + 1 = sr - \deg(A) - g - \rhoalgo.$$
	Besides, it readily follows from the definition of $\rhoalgo$ that $sr - \deg(A) - g \geq 2\rhoalgo + 1$ and the proposition follows.
\end{proof}

From now on, we fix a divisor $A_1$ satisfying the requirements of Lemma~\ref{PARAM} (so, in particular, we assume Hypothesis~\ref{HYPO2}).

\subsubsection{Error localization}

We denote by $m=c+e$ the received message as the sum of a codeword $c \in \operatorname{LAG}_{\DD}(A,D)$ and an error $e=(e_1, \ldots, e_s)$ with $\wsrk(e) \leq \rhoalgo$.
The first step for decoding consists in finding a localizing function $f\in \DD^\star$ in the following sense.

\begin{Definition}
	An \emph{error localizing function} relative to the code $\operatorname{LAG}_{\DD}(A,D)$ and an error $e$ is a function $f \in \DD^\star$ such that
	$e_i^\star \circ \bar\varepsilon_i(f) = 0$ for all $i \in [s]$.
\end{Definition}

We note that the vanishing condition $e_i^\star \circ \bar\varepsilon_i(f) = 0$ is equivalent to the inclusion $\operatorname{im}{e_i} \subset \ker(\bar\varepsilon_i(f^\star))$, which justifies the terminology of ``localizing function''.

\begin{Proposition}\label{ERRLOCPROP}
	There exists a nonzero localizing function in $\Lambda_{\DD}(A_1)$.
\end{Proposition}
\begin{proof}
	It follows from the Riemann--Roch Theorem that
	$$\lambda_{\DD}(A_1) - \lambda_{\DD}(A_1,\ker e^\star) \leq r \cdot \dim_\FF(\ker e^\star) = r\cdot\wsrk(e).$$
	As $\lambda_{\DD}(A_1) > r \cdot \wsrk(e) $, we deduce that $\lambda_{\DD}(A_1, \ker e^\star) > 0$.
	We conclude by noticing that any function in $\Lambda_{\DD^\star}(A_1,\operatorname{im}(e))$ is an error localizing function.
\end{proof}

\begin{Proposition}\label{ERRLOC}
	A function $f \in\Lambda_{\DD^\star}(A_1)$ is an error localizing function if and only if
	\begin{equation}\label{LOC}
	\forall h \in \Lambda_{\DD^\star}({\can+\operatorname{CoNr}(D)-A}- A_1), \quad
	\langle\bar\varepsilon_{D}(fh),m\rangle=0.
	\end{equation}
\end{Proposition}
\begin{proof}
	Let $f \in \Lambda_{\DD^\star}(A_1)$ and
	set $A_2 := \can+\operatorname{CoNr}(D) - A - A_1$ for simplicity.
	Noticing that $\can$ and $\operatorname{CoNr}(D)$ are both Galois divisors, we conclude that $A_2$ is Galois too.
	By Lemma~\ref{LAGBELG1}, we thus have:
	$$\Lambda_{\DD^\star}(A_1) \cdot \Lambda_{\DD^\star}(A_2) \subset \Lambda_{\DD^\star}\big(\can+\operatorname{CoNr}(D) - A\big).$$
	It then follows from Theorem~\ref{NONCAN} that $\langle\bar\varepsilon_{D}(fh), c\rangle=0$ for all $h \in \Lambda_{\DD^\star}(A_2)$.
	Thus
	\begin{equation}
	\label{eq:pairingfgm}
	\langle\bar\varepsilon_{D}(fh), m\rangle = \langle\bar\varepsilon_{D}(fh), e\rangle
	= \sum_{i=0}^s \operatorname{Tr}_\FF \big(\bar\varepsilon_i(fh) \circ e_i^\star\big)
	= \sum_{i=0}^s \operatorname{Tr}_\FF \big(e_i^\star \circ \bar\varepsilon_i(f) \circ \bar\varepsilon_i(h)\big).
	\end{equation}
	If $f$ is a localizing function, then $e_i^\star \circ \bar\varepsilon_i(f)  = 0$, so $\langle\bar\varepsilon_{D}(fh), m\rangle$ vanishes as well.

	Conversely, let us assume that $\langle\bar\varepsilon_{D}(fh), m\rangle = 0$ for all $h \in \Lambda_{\DD^\star}(A_2)$.
	We define $W_i := \operatorname{im} e_i^\star$ and set $W := \bigoplus_{i=1}^s W_i$ accordingly.
	We notice that the right hand side of Equation~\eqref{eq:pairingfgm} can be rewritten as
	$\Phi_W\big(e_i^\star\circ \bar\varepsilon_i(f), \,  \bar\varepsilon_i(h)_{|W_i}\big)$.
	Our assumption then reads
	\begin{equation}\label{eq:phiW}
	\forall h \in \Lambda_{\DD^\star}(A_2), \quad \Phi_W\big(e_i^\star\circ \bar\varepsilon_i(f), \,  \bar\varepsilon_i(h)_{|W_i}\big) = 0.
	\end{equation}
	Observing moreover that
	\begin{align*}
	\operatorname{deg}({\can+\operatorname{CoNr}(D)-A}-A_1, W) 
	& = 2g - 2 + sr - \deg(A + A_1) - \dim_\FF W \\
	& > 2g - 2 + \rhoalgo - \wsrk(e) > 2g - 2
	\end{align*}
	we can apply Theorem~\ref{FUNDTHEO} to get the surjectivity of the map
	$$\bar\varepsilon_{D|W} : \Lambda_{\DD^\star}\big(\can+\operatorname{CoNr}(D)-A - A_1\big) \longrightarrow \prod_{i=1}^s \operatorname{Hom}_{\FF}(W_i,V_{i}).$$
	Therefore Equation~\eqref{eq:phiW} shows that $e_i^\star\circ \bar\varepsilon_i(f)$ is orthogonal (for $\Phi_W$) to the whole space $\prod_{i=1}^s \operatorname{Hom}_{\FF}(W_i,V_{i})$. Thanks to Lemma~\ref{PAIRING}, it then must vanish.
\end{proof}

\subsubsection{Syndrome equation}

It now only remains to solve the syndrome equation restricted to the localizing space previously determined.

\begin{Proposition}\label{SYNEQ}
Let $f$ be a nonzero localizing error and let $e' = (e'_1, \ldots, e'_s) \in \prod_{i=1}^{s}\operatorname{End}_{\FF}(V_{i})$ such that
\begin{enumerate}
\item for all $i \in [s]$, we have $\operatorname{im}(e'_i) \subset \ker \bar\varepsilon_i(f^\star)$,
\item for all $h \in \Lambda_{\DD^\star}\big(\can+\operatorname{CoNr}(D) - A \big)$, we have
$\langle \bar\varepsilon_{D}(h), e' \rangle =\langle \bar\varepsilon_{D}(h), m\rangle$.
\end{enumerate}
Then $e' = e$.
\end{Proposition}
\begin{proof}
	From Theorem~\ref{NONCAN}, we derive that the difference $e'-e$ lies in $\LAG_{\DD}(A,D)$.
	Besides its image lies in $\bigoplus_{i=1}^s \operatorname{im}(e_i)$ and thus $\wsrk(e'-e) \leq \wsrk(e) \leq \rhoalgo$.
	Since the minimum distance of $\LAG_{\DD}(A,D)$ is at least $\ddesign(A) > \rhoalgo$, we conclude that $e' = e$.
\end{proof}


Algorithm~\ref{ALGO} summarizes our decoding algorithm.

\begin{algorithm}\label{ALGO}
	\caption{Decoding algorithm for linearized Algebraic Geometry codes}
	\begin{algorithmic}[1]
		\Require{A message $m$ with an error of weight at most $\rhoalgo=\left \lfloor \frac{\ddesign(A)-g-1}{2}\right \rfloor$. }
		\Ensure{The error $e$ corresponding to the maximum likelihood principle.}
		\State Pick a divisor $A_1 \in \operatorname{Div}(\DD,D)$ satisfying the requirements of Lemma~\ref{PARAM}.
		\State Pick a nonzero solution $f\in \Lambda_{\DD^\star}(A_1)$ of the linear system
		$$\forall h \in \Lambda_{\DD^\star}\big(\can+\operatorname{CoNr}(D)-A-A_1\big), \quad \langle \bar\varepsilon_{D}(fh),m\rangle=0.$$
		\State Solve the linear system with $e \in \prod_{i=1}^{s}\operatorname{Hom}_{\FF}(V_{i},\ker \bar\varepsilon_i(f^\star))$:
		$$\forall h \in \Lambda_{\DD^\star}\big(\can+\operatorname{CoNr}(D)-A\big), \quad \langle \bar\varepsilon_{D}(h), e\rangle = \langle \bar\varepsilon_{D}(h), m \rangle.$$
		\State Return $e$.
	\end{algorithmic}
\end{algorithm}

\begin{Theorem}\label{LAGDECODING}
	We assume Hypothesis~\ref{HYPO} and~\ref{HYPO2}.
	Then, Algorithm \ref{ALGO} can decode any error of weight at most
	$$\rhoalgo = \left\lfloor \frac{\ddesign(A)-g-1}{2} \right\rfloor.$$
	Moreover, its complexity is polynomial in $s$ and $r$.
\end{Theorem}
\begin{proof}
	Correctness follows from Propositions~\ref{PARAM}, \ref{ERRLOCPROP} and~\ref{ERRLOC}. All steps clearly have polynomial complexity in $s$ and $r$, as the dimensions of both
	$\Lambda_{\DD^\star}(\can+\operatorname{CoNr}(D) - A)$ and $\Lambda_{\DD^\star}(\can+\operatorname{CoNr}(D)-A - A_1)$ are in $O(sr)$.
 \end{proof}

%

\section*{Acknowledgments}
This work was supported by ANR-21-CE39-0009-BARRACUDA and ANR-22-CPJ2-0047-01.

\appendix
\section{SageMath implementation of the decoding algorithm}\label{SAGE}

We provide some results from a SageMath implementation of our decoding algorithm.
The source code is available at \url{https://plmlab.math.cnrs.fr/drain/lag-decoding.git}.
We set a Kummer, Artin--Schreier or Artin--Schreier-Witt extension $\LL/\KK$ and assume the place(s) at infinity to be totally ramified in $\LL/\KK$.
For simplicity, we also assume the evaluation places  $Q_{i,j}|P_i$ to be totally split.
The code $\operatorname{LAG}_{\DD}(A,D)$ is the image of the following morphisms up to conjugation in a basis corresponding to the fiber $(Q_{i,0},\ldots,Q_{i,r-1})$ above $P_i$:	
$$\begin{array}{rcl}
\varepsilon'_{i}: \quad \Lambda_{\DD}(A)& \rightarrow & \mathbb{F}_q^{r\times r}\\
	 \lambda,T &\mapsto & \begin{bmatrix}\pi_{{i,0}}(\lambda)& 0 & \dots & 0 & 0 \\0 & \pi_{{i,1}}(\lambda) & \dots & 0 & 0 \\0 & 0 & \ddots & & 0 \\ \vdots & & \ddots & 0 & \vdots \\0 & 0 & \dots & 0 & \pi_{{i,r-1}}(\lambda)\end{bmatrix}, \begin{bmatrix}0 & 0 & \dots & 0 & \pi_{{i,0}}(x) \\1 & 0 & \dots & 0 & 0 \\0 & 1 & \ddots & & 0 \\ \vdots & & \ddots & 0 & \vdots \\0 & 0 & \dots & 1 & 0 \end{bmatrix}
\end{array}$$
where $\pi_{i,j}$ denotes the projection $\mathcal{O}_{Q_{i,j}}\rightarrow \mathcal{O}_{Q_{i,j}}/Q_{i,j}$.
These morphisms $\varepsilon'_{i}$ extend to the gauge order $\Lambda_{P_i}$ and are conjugated to the evaluation morphisms $\bar\varepsilon_i$.
Indeed, we check $$\varepsilon'_{i}(-) = \varepsilon_{i}(u_{P_i}^{-1}\cdot -\cdot u_{P_i} ). $$
\subsection*{SageMath use case}

We present hereunder a decoding test for a Kummer extension $\LL/\KK$ of degree $2$ over another Kummer extension $\KK/\mathbb{F}_{81}(t)$ of degree $2$.
{\small
\begin{sagecommandline}
    sage: load('LAG_decoding.sage')
    sage: F = GF(81)
    sage: DLx = DivisionAlgebra(F=F, rK=2, r=2, pol_K_arg=[0,2,1,1], pol_L_arg=[0,1], pol_x_arg=[F.gen(),1]) 
    sage: LAG = LinearizedAGCode(DLx=DLx, degA=7, s_eval=7) 
    sage: LAG
    sage: omega = LAG.maximal_decodable_error_weight() 
    sage: omega 
    sage: error = LAG.random_error(omega)
    sage: error  # RANDOM ERROR 
    sage: word = LAG.random_codeword()
    sage: message = word + error 
    sage: message # RANDOM MESSAGE 
    sage: decoded_error = LAG.decode_msg(message) # DECODING ...
    sage: decoded_error == error # CHECKING THE RESULT
\end{sagecommandline}
}

\subsection*{SageMath benchmark}

We tested the decoding of linearized AG codes on different Kummer and Artin--Schreier--Witt extensions.
A benchmark is reported on Figure~\ref{fig:s2} (resp.~Figure~\ref{fig:s3}) for the initial computation of Riemann--Roch spaces (resp.~iterative message decoding corresponding to codes described in Figure~\ref{fig:s1}).
The computations were run on a computer with AMD Ryzen 7 PRO 8840U (3,30 GHz up to 5,10 GHz) Processor and 16 Go DDR5-5 600MT/s (SODIMM).

\begin{figure}[t]
        \hfill\begin{tabular}{ |c|c|c|c|c|c|c|c|c|}            
            \hline
             &  $\FF$ & $\KK$   &  $\LL$& $x_{\DD}$  & $g_{\KK}$ &$g_{\LL}$ &  $g_{\DD}$&  $s$    \\
            \hline
           $C(2)$          &   $\mathbb{F}_{5^2}$ &{$\mathbb{F}_{5^2}(t)$}  &   {$\frac{\mathbb{F}_{5^2}(t)[y]}{y^2 + 4 t^3 + 4 t^2 + 4 t}$}& $t + z_2$   &  $0$&$1$&$3/2$ & $14$ \\
            \hline
			$C(3)$       &  $\mathbb{F}_{3^3}$ &{$\mathbb{F}_{3^3}(t)$}   &  {$\frac{\mathbb{F}_{3^3}(t)[y]}{y^3 + 2 y + 2 t^3 + 2 t^2 + 2 t}$} & $t + z_3$ &   $0$&$1$&$2$  & $9$ \\
            \hline
           $C(4)$         &    $\mathbb{F}_{5^2}$ &{$\mathbb{F}_{5^2}(t)$}   &   {$\frac{\mathbb{F}_{5^2}(t)[y]}{y^4 + 4 t^3 + 4 t^2 + 4 t}$} & $t + z_2 $ &   $0$&$3$&$9/2$  & $7$\\
            \hline
          $C(5)$          &     $\mathbb{F}_{5^2}$  & {$\mathbb{F}_{5^2}(t)$}    &  {$\frac{\mathbb{F}_{5^2}(t)[y]}{y^5 + 4 y + 4 t^3 + 4 t^2 + 4 t}$} & $t + z_2$ &   $0$&$4$&$6$  & $7$\\
            \hline
          $C(8)$          &     $\mathbb{F}_{7^2}$  & {$\mathbb{F}_{7^2}(t)$}   &  {$\frac{\mathbb{F}_{7^2}(t)[y]}{y^8 + 6 t^3 + 6 t^2 + 6 t}$} & $t + z_2$  &    $0$&$7$&$21/2$& $6$\\
            \hline
         $C(2,2)$    &      $\mathbb{F}_{3^4}$ & { $\frac{\mathbb{F}_{3^4}(t)[x]}{x^2 + 2 t^3 + 2 t^2 + t}$} &   {$\frac{\KK(x)[y]}{y^2 + 2 x}$} & $t + z_4$ &   $1$&$3$&$4$  & $15$ \\
            \hline
        \end{tabular}\hfill\null
\caption{Description of the tested codes, where $z_n$ denotes a generator of $\mathbb{F}^\times_{p^n}$  }\label{fig:s1}
\end{figure}
\begin{figure}[t]
        \hfill\begin{tabular}{ |r|r|r|r||r|r|r||r|r|r|}            
            \hline
             &  { $\omega$ }    & $\operatorname{deg}A$ & time &  { $\omega$ }   & $\operatorname{deg}A$   & time &  { $\omega$  }   &  $\operatorname{deg}A$ & time    \\
            \hline
           $C(2)$           &    $1$ &$23$ &$1.985$   &     $6$ & $12$ & $2.978$  &   $11$ &$2$& $4.003$ \\
            \hline
			$C(3)$        & $1$ &$22$ &$4.697$    &     $6$ &$12$ &$5.996$  &   $10$ &$3$& $7.756$  \\
            \hline
           $C(4)$         &     $1$ &$20$ &$7.023$   &     $4$ & $14$&$7.682$  &   $7$ &$8$& $8.793$  \\
            \hline
          $C(5)$         &       $1$  & $26$& $21.415$   &      $5$ &$18$& $23.994$  &   $8$ &$11$& $27.936$  \\
            \hline
          $C(8)$        &         $1$  & $34$ &$142.395$  &      $4$ &$27$ &$153.397$  &   $8$ &$20$& $164.652$ \\
            \hline
         $C(2,2)$    &       $1$ &  $23$& $15.738$   &$15$ &   $15$ & $23.138$&   $9$ &$7$& $30.225$  \\
            \hline
        \end{tabular}\hfill\null
\caption{CPU Timings in seconds for constructing a code.}\label{fig:s2}
\end{figure}
\begin{figure}[t]
        \hfill\begin{tabular}{ |r|r|r|r||r|r|r||r|r|r|}           
   \hline
             &    $\omega$ &  $\operatorname{deg}A$   & time &   $\omega$ &  $\operatorname{deg}A$  & time &  $\omega$ &  $\operatorname{deg}A$ & time    \\
            \hline
           $C(2)$  &   $1$  &$23$  &  $0.009$    &    $6$ & $12$ & $0.026$       &   $11$  & $2$ & $0.052$ \\
            \hline
           $C(3)$    &    $1$  &$22$   & $0.016$   &     $6$ &$12$  & $0.045$     &  $10$ & $3$   & $0.094$  \\
            \hline
           $C(4)$   &   $1$   & $20$ & $0.032$   &    $4$   & $14$ & $0.065$      &  $7$ &  $8$ & $0.137$   \\
            \hline
           $C(5)$   &    $1$ &   $26$  & $0.093$  &     $5$ & $18$& $0.236$       &  $8$ &  $11$ & $0.440$   \\
            \hline
            $C(8)$   &  $1$   &$34$  & $1.386$  &     $4$ & $27$& $2.394$      &   $8$ &  $20$ & $4.741$ \\
            \hline
             $C(2,2)$   &   $1$  &  $23$ & $0.016$  &  $5$ &  $15$& $0.028$      &   $9$ &  $7$& $0.047$ \\
            \hline
        \end{tabular}\hfill\null
\caption{CPU Timings in seconds for decoding a single random message in a code.}\label{fig:s3}
\end{figure}

\bibliographystyle{alpha}
\bibliography{LAG-decoding}
\end{document}